\documentclass[UKenglish,runningheads,envcountsect,envcountsame]{llncs}
\usepackage{etex,etoolbox}
\usepackage{marvosym}



\usepackage{ifdraft}
\ifdraft{
\usepackage[nohead,nofoot,
            headsep = 1cm,
            footskip = 1cm,
            text={12.2cm,19.3cm},
            paperwidth=16.2cm,
            paperheight=23.3cm,
            marginparsep=1mm,
            marginparwidth=1.88cm,
            footskip=1cm,
            centering
            ]{geometry}
}{}

\makeatletter
\RequirePackage[bookmarks,unicode,colorlinks=true]{hyperref}%
   \def\@citecolor{black}%
   \def\@urlcolor{blue}%
   \def\@linkcolor{black}%

\def\orcidID#1{\smash{\href{http://orcid.org/#1}{\protect\raisebox{-1.25pt}{\protect\includegraphics{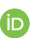}}}}}
\makeatother

\usepackage[notref,notcite,final]{showkeys}
\usepackage{cite}

\usepackage{seqsplit}
\usepackage{xstring}
\newcommand{\defaultshowkeysformat}[1]{%
\StrSubstitute{#1}{ }{\textvisiblespace}[\TEMP]%
\parbox[t]{\marginparwidth}{\raggedright\normalfont\small\ttfamily\(\{\){\color{red!50!black}\expandafter\seqsplit\expandafter{\TEMP}}\(\}\)}%
}

\renewcommand*\showkeyslabelformat[1]{%
\noexpandarg%
\defaultshowkeysformat{#1}%
}

\newcounter{blubber}

\makeatletter
\def\moverlay{\mathpalette\mov@rlay}
\def\mov@rlay#1#2{\leavevmode\vtop{%
   \baselineskip\z@skip \lineskiplimit-\maxdimen
   \ialign{\hfil$\m@th#1##$\hfil\cr#2\crcr}}}
\newcommand{\charfusion}[3][\mathord]{
    #1{\ifx#1\mathop\vphantom{#2}\fi
        \mathpalette\mov@rlay{#2\cr#3}
      }
    \ifx#1\mathop\expandafter\displaylimits\fi}
\makeatother

\newcommand{\size}{\mathsf{size}}

\newcommand{\target}{\chi}

\newcommand{\Cls}{\mathcal}
\newcommand{\CM}{{\Cls M}}
\newcommand{\FA}{{\mathfrak A}}

\newcommand{\Dist}{\mathcal{D}}
\newcommand{\Rat}{{\mathbb{Q}}}
\newcommand{\Nat}{{\mathbb{N}}}

\newcommand{\Bag}{\mathcal{B}}
\newcommand{\Gm}{\mathcal{G}}

\newcommand{\Set}{\mathsf{Set}}

\newcommand{\Sem}[1]{{[\![#1]\!]}}
\newcommand{\hearts}{\heartsuit}

\newcommand{\sem}[1]{[\![#1]\!]}

\newcommand{\Pow}{\mathcal{P}}

\usepackage{booktabs}   %% for formal tables:
\usepackage{xcolor}
\usepackage{soul}
\usepackage[utf8]{inputenc}
\usepackage{amsmath,amssymb}
\usepackage{stmaryrd} %\ll/rrbracket
\usepackage{todonotes}
\usepackage{xspace}

\renewcommand{\Box}{\square}
\renewcommand{\Diamond}{\lozenge}

\newcommand{\prios}{2nk}

\newcommand{\detcarrier}{D_\target}
\newcommand{\detprio}{\beta}

\newcommand{\parti}{\rightharpoonup}

\DeclareMathOperator{\LFP}{\mathsf{LFP}}
\DeclareMathOperator{\GFP}{\mathsf{GFP}}

\newcommand\NP{$\textsc{NP}$}
\newcommand\coNP{$\textsc{coNP}$\xspace}
\newcommand\UP{$\textsc{UP}$}
\newcommand\coUP{$\textsc{co-UP}$\xspace}

\usepackage[inline]{enumitem}
\setlist[enumerate,1]{label=(\arabic*),font=\normalfont,align=left,leftmargin=0pt,labelindent=0pt,listparindent=\parindent,labelwidth=0pt,itemindent=!,topsep=3pt,parsep=0pt,itemsep=3pt,start=1}
\setlist[enumerate,2]{label=(\alph*),font=\normalfont,labelindent=*,leftmargin=*,start=1}
\setlist[itemize]{labelindent=*,leftmargin=*,topsep=5pt,itemsep=3pt}
\setlist[description]{labelindent=*,leftmargin=*,itemindent=-1 em}

\usepackage{xspace}
\usepackage{xcolor}
\usepackage{soul}
\usepackage[utf8]{inputenc}
\usepackage{amsmath,amssymb}
\usepackage{stmaryrd} %\ll/rrbracket

\numberwithin{equation}{section}

\usepackage[author=anonymous,marginclue,footnote,draft]{fixme}
\FXRegisterAuthor{dh}{adh}{DH}
\FXRegisterAuthor{ls}{als}{LS}
\FXRegisterAuthor{sm}{asm}{SM}

\usepackage{algorithmicx}
    \usepackage[ruled]{algorithm}
\usepackage[noend]{algpseudocode}

\usepackage{tikz}
 \usetikzlibrary{trees}
 \usetikzlibrary{shapes}
 \usetikzlibrary{fit}
 \usetikzlibrary{shadows}
 \usetikzlibrary{backgrounds}
 \usetikzlibrary{arrows,automata}

\tikzset{
   n/.style= {circle,fill,inner sep=1.5pt,node distance=2cm}
  ,acc/.style={circle,draw,inner sep=3pt,node distance=2cm}
  ,phantom/.style={circle},
  ,arr/.style={->, >=stealth, semithick, shorten <= 3pt, shorten >= 3pt}
}

\renewcommand{\Box}{\square}
\renewcommand{\Diamond}{\lozenge}

\newcommand{\takeout}[1]{\empty}

\makeatletter
\newcommand\mysubsec{\@startsection{paragraph}{4}{\z@}%
  {-6\p@ \@plus -4\p@ \@minus -4\p@}%
  {-0.5em \@plus -0.22em \@minus -0.1em}%
  {\normalfont\normalsize\bfseries}}
\makeatother

\spnewtheorem{assumptions}[theorem]{Assumptions}{\bfseries}{\rmfamily}
\spnewtheorem{notation}[theorem]{Notation}{\bfseries}{\rmfamily}
\spnewtheorem{observation}[theorem]{Observation}{\bfseries}{\rmfamily}
\spnewtheorem{defn}[theorem]{Definition}{\bfseries}{\rmfamily}
\spnewtheorem{expl}[theorem]{Example}{\bfseries}{\rmfamily}
\spnewtheorem{rem}[theorem]{Remark}{\bfseries}{\rmfamily}
\spnewtheorem{fact}[theorem]{Fact}{\bfseries}{\rmfamily}
\spnewtheorem{construction}[theorem]{Construction}{\bfseries}{\rmfamily}
\spnewtheorem{examples}[theorem]{Examples}{\bfseries}{\rmfamily}

\begin{document}

\title{Quasipolynomial Computation \\of Nested Fixpoints}
%\titlerunning{} % if needed
\author{%
  Daniel Hausmann (\Letter)\orcidID{0000-0002-0935-8602}
  \and
  Lutz Schröder (\Letter)\orcidID{0000-0002-3146-5906}~\thanks{Work forms part of the DFG-funded project CoMoC (SCHR 1118/15-1, MI 717/7-1).}%\thanks{Supported by~???}
%  \and
%  Aaron Strahlberger%\thanks{Supported by~???}
%  \and
%  Paula Welzenbach%\thanks{Supported by~???}
}
\authorrunning{D.~Hausmann, L.~Schröder}%, A. Strahlberger, P. Welzenbach}

\institute{Friedrich-Alexander-Universität
  Erlangen-Nürnberg, Erlangen, Germany\\
  \email{\{daniel.hausmann,lutz.schroeder\}@fau.de}%
}

\maketitle
\begin{abstract}
  It is well-known that the winning region of a parity game with $n$
  nodes and $k$ priorities can
  be %specified by a $\mu$-calculus formula with alternation depth $k$, or, equivalently,
  computed as a $k$-nested fixpoint of a suitable function;
  straightforward computation of this nested fixpoint requires
  $\mathcal{O}(n^{\frac{k}{2}})$ iterations of the function.  Calude
  et al.'s recent quasipolynomial-time parity game solving algorithm
  essentially shows how to compute the same fixpoint in only
  quasipolynomially many iterations by reducing parity games
  to quasipolynomially sized safety games. Universal graphs have been
  used to modularize this transformation of parity games to equivalent
  safety games that are obtained by combining the original
  game with a universal graph. 
   We show that this approach
  naturally generalizes to the computation of solutions of systems of
  \emph{any} fixpoint equations over finite lattices; hence,
  the solution of fixpoint equation systems can be computed 
  by quasipolynomially many iterations of the equations.
  We present applications to
  modal fixpoint logics and games beyond relational semantics. 
  For instance, the
  model checking problems for the energy $\mu$-calculus, finite latticed
  $\mu$-calculi, and
  the graded and the (two-valued)
  probabilistic $\mu$-calculus -- with numbers coded in binary -- can
  be solved via nested fixpoints of functions that differ
  substantially from the function for parity games but still can be
  computed in quasipolynomial time; our result hence implies that
  model checking for these $\mu$-calculi is in
  $\textsc{QP}$. Moreover, we improve the exponent in known
  exponential bounds on satisfiability
  checking. % A second implication
%   of our result lies in satisfiability checking for generalized
%   $\mu$-calculi, including the graded, probabilistic and
%   alternating-time variants; in a general setting that covers all the
%   mentioned cases, our result immediately improves the upper time
%   bound for satisfiability checking for fixpoint formulae of size~$n$
%   with alternation-depth $k$ from $2^{\mathcal{O}({n^2k^2\log n})}$ to
%   $2^{\mathcal{O}({nk\log
%       n})}$. % , under mild assumptions on the complexity of the so-called one-step satisfiability problem
% %       of the logic at hand.
% %       As a side-result we show that the problem of computing
% %       $k$-nested fixpoints also is in
% %       $\textsc{NP}\cap\textsc{coNP}$,
% %       underlining the close connection between solving parity
% %       games and
% %       computing alternating fixpoints.
\end{abstract}
\medskip

\section{Introduction}

Fixpoints are pervasive in computer science, governing large portions
of recursion theory, concurrency theory, logic, and game theory. One
famous example are parity games, which are central, e.g., to networks
and infinite processes~\cite{BodlaenderEA01}, tree
automata~\cite{Zielonka98}, and $\mu$-calculus model
checking~\cite{EmersonEA01}. Winning regions in parity games can be
expressed as nested fixpoints of particular set functions
(e.g.~\cite{DawarGraedel08,BruseEA14}). In recent breakthrough work on
the solution of parity games in quasipolynomial time, Calude et
al.~\cite{CaludeEA17} essentially show how to compute this particular
fixpoint in quasipolynomial time, that is, in time
$2^{\mathcal{O}({(\log n)^c})}$ for some constant $c$. 
Subsequently, it has been 
shown~\cite{JurdzinskiLazic17,CzerwinskiEA19,ColcombetFijalkow19} 
that universal graphs (that is, even graphs into which every even 
graph of a certain size embeds by a graph morphism)
can be used to transform parity games to equivalent safety
games obtained by pairing the original game with a universal
graph; the size of these safety games is determined by the size of
the employed universal graphs and it has been shown~\cite{CzerwinskiEA19,ColcombetFijalkow19} that there
are universal graphs of quasipolynomial size. This yields a uniform
algorithm for solving parity games to which all currently known 
quasipolynomial algorithms for parity games have been shown to
instantiate using appropriately defined universal graphs~\cite{CzerwinskiEA19,ColcombetFijalkow19}.

Briefly, our contribution in the
present work is to show that the method of using universal graphs
to solve parity games generalizes to the computation of nested
fixpoints of arbitrary functions over finite lattices. 
That is, given functions
$f_i:L^{k+1}\to L$, $0\leq i \leq k$ on a finite lattice~$L$,
we give an algorithm that uses
universal graphs to compute the solutions of systems of equations
\begin{equation*}
  X_i=_{\eta_i} f_i(X_0,\ldots,X_k) \qquad\qquad 0\leq i\leq k
\end{equation*}
where $\eta_i=\GFP$ (greatest fixpoint) or
$\eta_i=\LFP$ (least fixpoint). Since there are universal graphs
of quasipolynomial size, the
algorithm requires only quasipolynomially many iterations of the functions $f_i$ and hence runs in quasipolynomial time,
provided that all~$f_i$ are computable in quasipolynomial time.
While it seems plausible that this time bound
may also be obtained by translating  equation systems
to equivalent standard 
parity games by emulating Turing machines to encode the functions $f_i$ as Boolean circuits (leading to many additional states but avoiding exponential blowup during the process),
we emphasize that the main point of our result is not so much
the ensuing time bound but rather the insight that universal graphs
and hence many algorithms for
parity games can be used on a much more general level which yields
  a precise (and relatively low) quasipolynomial bound on the number of function 
  calls that are required to obtain solutions of fixpoint equation systems.

In more detail, the method of Calude et al.\ can be described as
annotating nodes of a parity game with histories of quasipolynomial
size and then solving this annotated game, but with a safety winning
condition instead of the much more involved parity winning
condition. It has been shown that these histories can be seen as nodes
in universal graphs, in a more general reduction of parity games to
safety games in which nodes from the parity game are annotated with
nodes from a universal graph.  This method has also been described as
pairing \emph{separating automata} with safety
games~\cite{CzerwinskiEA19}.  It has been shown~\cite{CzerwinskiEA19,ColcombetFijalkow19}
that there are exponentially sized universal graphs (essentially
yielding the basis for e.g. the
fixpoint iteration algorithm~\cite{BruseEA14} or the small progress
measures algorithm~\cite{Jurdzinski00}) and quasipolynomially sized
universal graphs (corresponding, e.g., to the succinct progress
measure algorithm~\cite{JurdzinskiLazic17}, or to the recent
quasipolynomial variant of Zielonka's algorithm~\cite{Parys19}).

Hasuo et al.~\cite{HasuoEA16}, and more generally, Baldan et
al.~\cite{BaldanEA19} show that nested fixpoints in highly general
settings can be computed by a technique based on progress measures,
implicitly using exponentially sized universal graphs, obtaining an
exponential bound on the number of iterations. Our technique is based
on showing that one can make explicit use of universal graphs,
correspondingly obtaining a quasipolynomial upper bound on the number
of iterations. In both cases, computation of the nested fixpoint is
reduced to a single (least or greatest depending on exact formulation)
fixpoint of a function that extends the given set function to keep
track of the exponential and quasipolynomial histories, respectively,
in analogy to the previous reduction of parity games to safety games.
Our central result can then be phrased as saying that the method of
transforming parity conditions to safety conditions using universal
graphs generalizes from solving parity games to solving systems of
equations that use arbitrary functions over finite lattices.  We use
\emph{fixpoint games}~\cite{Venema08,BaldanEA19} to obtain the crucial
result that the solutions of equation systems have history-free
witnesses, in analogy to history-freeness of winning strategies in
parity games. These fixpoint games have exponential size but we show
how to extract \emph{polynomial-size} witnesses for winning strategies
of $\mathsf{Eloise}$, and use these witnesses to show that any node
won by $\mathsf{Eloise}$ is also won in the safety game obtained by a
universal graph.  For the backwards direction, we show that a witness
for satisfaction of the safety condition regarding the universal graph
induces a winning strategy in the fixpoint game. This proves that
universal graphs can be used to compute nested fixpoints of arbitrary
functions over finite lattices and hence yields the quasipolynomial
upper bound for computation of nested fixpoints.  Moreover, we present
a progress measure algorithm that uses the nodes of a quasipolynomial
universal graph to measure progress and that can be used to
efficiently compute nested fixpoints of arbitrary functions over
finite lattices.
 % As a side result, we
% moreover use our polynomial-sized witnesses for containment in nested
% fixpoints to show that computing nested fixpoints of functions that
% can be computed in polynomial time is also in \NP$\,\cap$ \coNP (a
% bound that is presently imcomparable to the $\textsc{QP}$ bound).

As an immediate application of these results, we improve known
deterministic algorithms for solving \emph{energy parity games}
\cite{ChatterjeeDoyen12}, that is,
parity games in
which edges have additional integer weights and for which
the winning condition is a combined parity condition and
a (quantitative) positivity condition on the 
sum of the accumulated weights. Our results also show that
the model checking problem for the associated
\emph{energy $\mu$-calculus}\cite{AmramEA20}
is in $\textsc{QP}$. In a similar fashion, we obtain quasipolynomial algorithms
for model checking in
latticed $\mu$-calculi\cite{BrunsGodefroid04} in which the
truth values of formulae are computed over arbitrary finite lattices,
and for solving associated latticed parity games~\cite{KupfermanLustig07}.

Furthermore, our results improve generic upper
complexity bounds on model checking and satisfiability checking in the
\emph{coalgebraic $\mu$-calculus}~\cite{CirsteaEA11a}, which serves as
a generic framework for fixpoint logics beyond relational
semantics. Well-known instances of the coalgebraic $\mu$-calculus
include the alternating-time $\mu$-calculus~\cite{AlurEA02}, the
graded $\mu$-calculus~\cite{KupfermanEA02}, the (two-valued)
probabilistic $\mu$-calculus~\cite{CirsteaEA11a,LiuEA15}, and the
monotone $\mu$-calculus~\cite{EnqvistEA15} (the ambient fixpoint logic
of concurrent dynamic logic CPDL~\cite{Peleg87} and Parikh's game
logic~\cite{Parikh85}). This level of generality is achieved by
abstracting system types as set functors and systems as coalgebras
for the given functor following the paradigm of universal
coalgebra~\cite{Rutten00}.
It was previously shown~\cite{HausmannSchroder19b} that the model
checking problem for coalgebraic $\mu$-calculi reduces to the
computation of a nested fixpoint. This fixpoint may be seen as a
coalgebraic generalization of a parity game winning region but can be
literally phrased in terms of small standard parity games (implying
quasipolynomial run time) only in restricted cases. Our results show
that the relevant nested fixpoint can be computed in quasipolynomial
time in all cases of interest. Notably, we thus obtain as new specific
upper bounds that even under binary coding of numbers, the model
checking problems of both the graded $\mu$-calculus and the
probabilistic $\mu$-calculus are in $\textsc{QP}$, even when the syntax is
extended to allow for (monotone) polynomial inequalities.

Similarly, the satisfiability problem of the coalgebraic
$\mu$-calculus has been reduced to a computation of a nested
fixpoint~\cite{HausmannSchroder19a}, and our present results imply a
marked improvement in the exponent of the associated exponential time
bound.  Specifically, the nesting depth of the relevant fixpoint is
exponentially smaller than the basis of the lattice. Our results imply that this fixpoint is computable in
polynomial time so that the complexity of satisfiability checking in
coalgebraic $\mu$-calculi 
drops from $2^{\mathcal{O}({n^2k^2\log n})}$ to
$2^{\mathcal{O}({nk\log n})}$ for formulae of size $n$ and with
alternation depth $k$.

\paragraph{Related Work} The quasipolynomial bound on parity game
solving has in the meantime been realized by a number of alternative
algorithms. For instance, Jurdzinski and
Lazic~\cite{JurdzinskiLazic17} use succinct progress measures to
improve to quasilinear (instead of quasipolynomial) space; Fearnley et
al.~\cite{FearnleyEA19} similarly achieve quasilinear space.
Lehtinen~\cite{Lehtinen18} and Boker and
Lehtinen~\cite{BokerLehtinen18} present a quasipolynomial algorithm
using register games. Parys~\cite{Parys19} improves Zielonka's
algorithm~\cite{Zielonka98} to run in quasipolynomial time. In
particular the last algorithm is of interest as an additional
candidate for generalization to nested fixpoints, due to the known
good performance of Zielonka's algorithm in practice.  Daviaud et
al.~\cite{DaviaudEA18} generalize quasipolynomial-time parity game
solving by providing a pseudo-quasipolynomial algorithm for
mean-payoff parity games. On the other hand, Czerwinski et
al.~\cite{CzerwinskiEA19} give a quasipolynomial lower bound on
universal trees, implying a barrier for prospective polynomial-time
parity game solving algorithms. Chatterjee et
al.~\cite{ChatterjeeEA18} describe a quasipolynomial time
set-based symbolic algorithm for parity
game solving that is parametric in a \emph{lift} function that
determines how ranks of nodes depend on the ranks of their successors,
and thereby unifies the complexity and correctness analysis of various
parity game algorithms. Although part of the parity game structure is
encapsulated in a set operator $\mathit{CPre}$, the development is
tied to standard parity games, e.g.\ in the definition of the
$\mathit{best}$ function, which picks minimal or maximal ranks of
successors depending on whether a node belongs to $\mathsf{Abelard}$
or $\mathsf{Eloise}$.

Early work on the computation of unrestricted nested fixpoints has
shown that greatest fixpoints require less effort in the fixpoint
iteration algorithm, which can hence be optimized to compute nested
fixpoints with just $\mathcal{O}(n^{\frac{k}{2}})$ calls of the
functions at hand~\cite{LongEA94,Seidl96}, improving the previously
known (straightforward) bound $\mathcal{O}(n^k)$; here, $n$ denotes
the size of the basis of the lattice and $k$ the number of fixpoint
operators.
Recent progress in the field has established the above-mentioned
approaches using progress measures~\cite{HasuoEA16} and fixpoint
games~\cite{BaldanEA19} in general settings, both with a view to
applications in coalgebraic model checking like in the present
paper. In comparison to the present work, the respective bounds on the
required number of function iterations in the above unrestricted
approaches all are exponential.

A preprint of our present results, specifically the quasipolynomial
upper bound on function iteration in fixpoint computation, has been
available as an arXiv preprint for some
time~\cite{HausmannSchroeder19}. Subsequent to this preprint, Arnold,
Niwinski and Parys~\cite{ArnoldEA21} have improved the actual run time
by reducing the overhead incurred per iteration (and they give a form
of quasipolynomial lower bound for universal-tree-based algorithms),
working (like~\cite{HausmannSchroeder19}) in the less general setting
of directly nested fixpoints over powerset lattices; we show in
Section~\ref{sec:lifting} how such an improvement can be incorporated
also in our lattice-based algorithm.

\section{Notation and Preliminaries}
Let $U$ and $V$ be sets, and let $R\subseteq U\times U$ be a binary
relation on $U$. For $u\in U$, we then put
$R(u):=\{v\in U\mid (u,v)\in R\}$. We put $[k]=\{0,\ldots,k\}$ for
$k\in\mathbb{N}$.
\emph{Labelled graphs} $G=(W,R)$
consist of a set $W$ together with a relation
$R\subseteq W\times A \times W$ where $A$ is some set of labels;
typically, we use $A=[k]$ for some $k\in\mathbb{N}$.  An
\emph{$R$-path} in a labelled graph is a finite or infinite sequence
$v_0,a_0,v_1,a_1,v_2\dots$ (ending in a node from~$W$ if finite) such
that $(v_i , a_i,  v_{i+1})\in R$ for all~$i$. For $v\in W$ and $a\in A$,
we put $R_a(v)=\{w\in W\mid (v,a,w)\in R\}$ and sometimes write $|G|$
to refer to~$|W|$.  As usual, we write $U^*$ and $U^\omega$ for the
sets of finite sequences or infinite sequences, respectively, of
elements of~$U$.  The \emph{domain} $\mathsf{dom}(f)$ of a partial
function $f:U\parti V$ is the set of elements on which~$f$ is defined.
As usual, the \emph{(forward) image} of $A'\subseteq A$ under a function
$f:A\to B$ is $f[A']=\{b\in B\mid \exists a\in A'.\, f(a)=b\}$
and the \emph{preimage} $f^{-1}[B']$ of $B'\subseteq B$ under $f$ is defined
by $f^{-1}[B']=\{a\in A\mid \exists b\in B'.\, f(a)=b\}$.
\emph{Projections} $\pi_j:A_1\times \ldots\times A_m\to A_j$ for
$1\leq j\leq m$ are given by $\pi_i(a_1,\ldots,a_m)=a_j$.
We often regard (finite) sequences
$\tau=u_0,u_1,\ldots\in U^*\cup U^\omega$ of elements of $U$ as
partial functions of type $\mathbb{N}\parti U$ and then write
$\tau(i)$ to denote the element $u_i$, for $i\in\mathsf{dom}(\tau)$.
For $\tau\in U^*\cup U^\omega$, we define the set
\begin{math}
\mathsf{Inf}(\tau)=\{u\in U\mid \forall
i\geq 0.\,\exists j>i.\,\tau(j)=u\}
\end{math}
of elements that occur infinitely often in $\tau$ (so
$\mathsf{Inf}(\tau)=\emptyset$ for $\tau\in U^*$). An infinite
$R$-path $v_0,p_0,v_1,p_1,\dots$ in a labelled graph $G=(W,R)$ with
labels from $[k]$ is \emph{even} if
$\max(\mathsf{Inf}(p_0,p_1,\dots))$ is even, and~$G$ is \emph{even} if
every infinite $R$-path in~$G$ is even. We write $\Pow(U)$ for the
powerset of $U$, and $U^m$ for the $m$-fold Cartesian product
$U\times\dots\times U$. 
\mysubsec{Finite Lattices and Fixpoints} 
A \emph{finite lattice} $(L,\sqsubseteq)$ (often written just as $L$)
 consists of a non-empty finite set
$L$ together with a partial order $\sqsubseteq$ on $L$, such that
there is, for all subsets $X\subseteq L$, a join $\bigsqcup X$ and
a meet $\bigsqcap X$. The least and greatest elements of $L$ are defined
as $\top=\bigsqcup\emptyset$ and $\bot=\bigsqcap\emptyset$,
respectively.
A set $B_L\subseteq L$ such that $l=\bigsqcup\{b\in B_L\mid b\sqsubseteq l\}$
for all $l\in L$ is a \emph{basis} of $L$.
Given a finite lattice $L$, a function $g:L^k\to L$  is
\emph{monotone} if $g(V_1,\ldots, V_k)\sqsubseteq g(W_1,\ldots,W_k)$
whenever $V_i\sqsubseteq W_i$ for all $1\leq i\leq k$.  For monotone
$f:L\to L$, we put
\begin{align*}
\GFP f=&\textstyle\bigsqcup\{V\sqsubseteq L\mid V\sqsubseteq f(V)\} &
\LFP f=&\textstyle\bigsqcap\{V\sqsubseteq L\mid f(V)\sqsubseteq V\},
\end{align*}
which, by the Knaster-Tarski fixpoint theorem, are the greatest and
the least fixpoint of $f$, respectively.  Furthermore, we define
$f^0(V)=V$ and $f^{m+1}(V)=f(f^m(V))$ for $m\geq 0$, $V\sqsubseteq L$;
since $L$ is finite, we have $\GFP f=f^{n}(\top)$ and $\LFP f=f^{n}(\bot)$
by Kleene's fixpoint theorem.  
Given a finite set $U$ and a natural number $n$, $(n^U,\sqsubseteq)$ is a finite lattice, where $n^U=\{f:U\to [n-1]\}$ denotes the function space
from $U$ to $[n-1]$ and $f\sqsubseteq g$ if and only if for all
$u\in U$, $f(u)\leq g(u)$. For $n=2$, we obtain the powerset lattice $(2^U,\subseteq)$, also denoted by $\Pow(U)$,
 with least and greatest elements $\emptyset$ and $U$, respectively,
and basis $\{\{u\}\mid u\in U\}$.

%Since we essentially prove two
%striking similarities between parity game solving and computing nested fixpoints, the central contributions of this work rely heavily on certain previous results on parity games which
%we hence also recall in this section.
%We also use the computation of winning regions in parity games
%as our prime example of a problem that can be solved in an
%elegant and natural way by computing nested fixpoints.
%Applications of parity game
%solving are ubiquitous in model checking and satisfiability %checking for the $\mu$-calculus and in LTL synthesis.
\mysubsec{Parity games} A \emph{parity game} $(V,E,\Omega)$ consists
of a set of \emph{nodes} $V$, a left-total relation
$E\subseteq V\times V$ of \emph{moves} encoding the rules of the game,
and a \emph{priority function} $\Omega:V\to\mathbb{N}$, which assigns
\emph{priorities} $\Omega(v)\in\mathbb{N}$ to nodes $v\in
V$. Moreover, each node belongs to exactly one of the two players
$\mathsf{Eloise}$ or $\mathsf{Abelard}$, where we denote the set of
$\mathsf{Eloise}$'s nodes by~$V_\exists$ and that of
$\mathsf{Abelard}$'s nodes by~$V_\forall$. A \emph{play}
$\rho\in V^\omega$ is an infinite sequence of nodes that follows the
rules of the game, that is, such that for all $i\geq 0$, we have
$(\rho(i),\rho(i+1))\in E$.  We say that an infinite play
$\rho=v_0,v_1,\ldots$ is \emph{even} if the largest priority that
occurs infinitely often in it (i.e.\
$\max (\mathsf{Inf}(\Omega\circ\rho))$) is even, and \emph{odd}
otherwise, and call this property the \emph{parity} of~$\rho$. Player
$\mathsf{Eloise}$ \emph{wins} exactly the even plays and player
$\mathsf{Abelard}$ \emph{wins} all other plays.  A
\emph{(history-free) $\mathsf{Eloise}$-strategy} $s:V_\exists\parti V$
is a partial function that assigns single moves $s(x)$ to
$\mathsf{Eloise}$-nodes $x\in\mathsf{dom}(s)$. Given an
$\mathsf{Eloise}$-strategy $s$, a play $\rho$ is an \emph{$s$-play} if
for all $i\in\mathsf{dom}(\rho)$ such that $\rho(i)\in V_\exists$, we
have $\rho(i+1)=s(\rho(i))$.  An $\mathsf{Eloise}$-strategy
\emph{wins} a node $v\in V$ if $\mathsf{Eloise}$ wins all $s$-plays
that start at $v$.  We have a dual notion of
$\mathsf{Abelard}$-strategies; \emph{solving} a parity game consists
in computing the \emph{winning regions} $\mathsf{win}_\exists$ and
$\mathsf{win}_\forall$ of the two players, that is, the sets of states
that they respectively win by some
strategy. % A parity game is \emph{alternating} if
%  $E[V_E]\subseteq V_\forall$ and $E[V_\forall]\subseteq V_\exists$,
%  that is, if all of $\mathsf{Eloise}$'s moves lead to
%  $\mathsf{Abelard}$-nodes and vice versa.

  It is known that solving parity games is in \NP$\,\cap\,$\coNP (and,
  more specifically, in \UP$\,\cap\,$\coUP).  Recently it has also
  been shown~\cite{CaludeEA17} that for parity games with~$n$ nodes
  and~$k$ priorities, $\mathsf{win}_\exists$ and
  $\mathsf{win}_\forall$ can be computed in quasipolynomial
  time~$\mathcal{O}(n^{\log k+6})$.  Another crucial property of
  parity games is that they are \emph{history-free
    determined}~\cite{GraedelThomasWilke02}, that is, that every node
  in a parity game is won by exactly one of the two players and then
  there is a history-free strategy for the respective player that wins
  the node.
%This is reflected by the central fact that
%the winning regions in parity games can be computed by \emph{fixpoint
%  iteration}, that is, each of these regions is the solution of a
%system of fixpoint equations (over a powerset lattice) as introduced in
%Section~\ref{sec:fixpoints}:

%\noindent The Lemma also implies determinacy of parity games since we have 
%$\mathsf{win}_\exists = \mathsf{E}^{\alpha_\exists}_k = \overline{ \mathsf{A}^{\alpha_\forall}_k}= \overline{ \mathsf{win}_\forall}$
%where the second equality is by Lemma~\ref{fact:fpduality}
%since $\alpha_\forall=\overline{\alpha_\exists}$
%(that is, $\alpha_\forall$ and $\alpha_\exists$
%are dual functions).\\

\section{Systems of Fixpoint Equations}\label{sec:fixpoints}

We now introduce our central notion, that is, systems of fixpoint
equations over a finite lattice.  Throughout, we fix a finite lattice 
$(L,\sqsubseteq)$ and a basis $B_L$ of $L$ such that $\bot\notin B_L$,
 and $k+1$ monotone
functions $f_i:L^{k+1}\to L$, $0\leq i\leq k$.
\begin{defn}
  A \emph{system of equations} consists of $k+1$ equations of
  the form
\begin{align*}
X_i=_{\eta_{i}} f_i(X_0,\ldots,X_k)
\end{align*}
where $\eta_i\in\{\LFP,\GFP\}$, briefly referred to as~$f$.  For a
partial valuation $\sigma:[k]\rightharpoonup L$, we
inductively define
\begin{align*}
\sem{X_i}^\sigma= \eta_i X_i. f_i^\sigma,
\end{align*}
where the function $f_i^\sigma$ is given by
\begin{align*}
f_i^\sigma(A)
=f_i(&\sem{X_0}^{\sigma'},\ldots,\sem{X_{i-1}}^{\sigma'},A,
\mathsf{ev}(\sigma',i+1),\ldots,\mathsf{ev}(\sigma',k))
\end{align*}
for $A\in L$, where 
$(\sigma[i\mapsto A])(j)=\sigma(j)$ for
$j\neq i$ and $(\sigma[i\mapsto A])(i)=A$, $\sigma'=\sigma[i\mapsto A]$
and where
$\mathsf{ev}(\sigma,j)=\sigma(j)$ if $j\in\mathsf{dom}(\sigma)$ and
$\mathsf{ev}(\sigma,j)=\sem{X_j}^\sigma$ otherwise (the latter clause
handles \emph{free variables}).  Then, the \emph{solution} of the
system of equations is ~$\sem{X_k}^\epsilon$ where
$\epsilon:[k]\rightharpoonup L$ denotes the empty
valuation (i.e.\ $\mathsf{dom}(\epsilon)=\emptyset$).  Similarly, we
can obtain solutions for the other components as $\sem{X_i}^\epsilon$
for $0\leq i<k$; we drop the valuation index if no confusion arises,
and sometimes write $\sem{X_i}_f$ to make the equation system~$f$
explicit.  We denote by $\mathsf{E}^{f_0}$ the solution $\sem{X_k}$
for the \emph{canonical system of equations} of the particular shape
\begin{align*}
X_i&=_{\eta_i} X_{i-1}  &
X_0&=_{\GFP} f_0(X_0,\ldots,X_k),
\end{align*} 
where $0<i\leq k$, $\eta_i=\LFP$ for odd $i$ and $\eta_i=\GFP$ for even $i$.
\end{defn}

\begin{expl}\label{exmp:eqs}
\begin{enumerate}
\item \emph{Parity games and the modal $\mu$-calculus:} Let
  $(V,E,\Omega)$ be a parity game with priorities $0$ to $k$, take
  $L=\Pow(V)$, and consider the canonical system of fixpoint equations
  $\mathsf{E}^{f_\exists}$ for the function
%\quad\text{ and }\quad
%\mathsf{win}_\forall = \mathsf{A}^{f_\forall},
$f_\exists\colon \Pow(V)^{k+1}\to\Pow(V)$ %and
%$f_\forall:\Pow(V)^{k+1}\to\Pow(V)$
given by
\begin{align*}
f_\exists(V_0,\ldots,V_k)=&
\{v\in V_\exists\mid E(v)\cap V_{\Omega(v)}\neq\emptyset
\}\,\cup \{v\in V_\forall\mid  E(v)\subseteq V_{\Omega(v)},
\}
\end{align*}
for $(V_0,\ldots,V_k)\in\Pow(V)^{k+1}$.  It is well known that
$\mathsf{win}_\exists = \mathsf{E}^{f_\exists}$, i.e.\ parity games
can be solved by solving fixpoint equation systems.
%and
%\begin{align*}
%f_\forall(V_0,\ldots,V_k)=&
%\{v\in V_\exists\mid \exists 0\leq i\leq k.\,\Omega(v)=
%i, \\
%&\qquad\qquad\qquad\qquad\quad E(v)\subseteq V_i
%\}\,\cup\\
%&\{v\in V_\forall\mid \exists 0\leq i\leq k.\,\Omega(v)=
%i, \\
%&\qquad\qquad\qquad\qquad\quad E(v)\cap V_i\neq\emptyset
%\}.
%\end{align*}
Intuitively, $v\in f_\exists(V_0,\ldots,V_k)$ iff $\mathsf{Eloise}$
can enforce that some node in~$V_{\Omega(v)}$ is reached in the next
step. The nested fixpoint expressed by $\mathsf{E}^{f_\exists}$ (in
which least (greatest) fixpoints correspond to odd (even) priorities)
is constructed in such a way that $\mathsf{Eloise}$ only has to rely
infinitely often on an argument $V_i$ for odd $i$ if she can also
ensure that some argument $V_{j}$ for $j>i$ is used infinitely often.

Model checking for the \emph{modal $\mu$-calculus}~\cite{Kozen83} 
and solving parity games are linear-time equivalent problems. 
Formulae of the $\mu$-calculus are evaluated over Kripke frames
$(U,R)$ with set of states $U$ and transition relation $R$.
Formulae $\phi$ of the $\mu$-calculus can be directly represented
as equation systems over the lattice $\Pow(U)$ by 
recursively translating $\phi$ to equations, mapping
subformulae $\mu X_i.\,\psi(X_0,\ldots, X_k)$ and
$\nu X_j.\,\chi(X_0,\ldots, X_k)$ to equations
\begin{align*}
X_i &=_\mu \psi(X_0,\ldots, X_k)& X_j &=_\nu \chi(X_0,\ldots, X_k),
\end{align*}
and interpreting the modalities $\Diamond$ and $\Box$ by functions
\begin{align*}
f_\Diamond(X) & = \{u\in U\mid R(u)\cap X\neq \emptyset\} &
f_\Box(X) & = \{u\in U\mid  R(u)\subseteq X\}
\end{align*}
The solution of the resulting system of equations then is the
truth set of the formula $\phi$, that is,
model checking for the model $\mu$-calculus reduces to
solving fixpoint equation systems.
Furthermore, satisfiability checking for the modal $\mu$-calculus
can be reduced to solving so-called 
\emph{satisfiability games}~\cite{FriedmannLange13a}, that
is, parity games that are played over the set of states of a 
determinized parity automaton.
These satisfiability games can be expressed as systems of fixpoint 
equations, where the functions track transitions in the determinized automaton.

\item \emph{Energy parity games and the energy $\mu$-calculus:}
Energy parity games~\cite{ChatterjeeDoyen12} are two-player games
played over weighted game arenas $(V,E,w,\Omega)$, where
$w:E\to\mathbb{Z}$ assigns integer weights to edges. The winning
condition is the combination of a parity condition with
a (quantitative) positivity condition on the 
sum of the accumulated weights. It has been 
shown~\cite{ChatterjeeDoyen12,AmramEA20}, that 
$b=n\cdot d \cdot W$ is a sufficient upper bound on energy level accumulations
in energy parity games with $n$ nodes, $k$ priorities and maximum absolute
weight $W$. We define a function 
$f^\mathsf{e}_\exists:((b+1)^{V})^{k+1}\to (b+1)^{V}$
over the finite lattice $(b+1)^V$ (whose elements are functions
from $V$ to the set $\{0,\ldots,b+1\}$) by putting
\begin{align*}
(f^\mathsf{e}_\exists(V_0,\ldots,V_k))(v)=
\begin{cases}
\min (\mathsf{en}(v,V_{\Omega(v)})) & \text{ if } v\in V_\exists\\
\max (\mathsf{en}(v,V_{\Omega(v)})) & \text{ if } v\in V_\forall,
\end{cases}
\end{align*}
for 
$(V_0,\ldots,V_k)\in ((b+1)^V)^{k+1}$ and $v\in V$, 
using $\mathsf{en}(v,\sigma)$ as abbreviation for
\begin{align*}
\mathsf{en}(v,\sigma)=&\,\{n\in\{0,\ldots,b\}\mid
\exists u\in E(v).\,n=\max\{0,\sigma(u)-w(v,u)\}\}\,\cup\\
&\,\{b+1\mid
\exists u\in E(v).\,\sigma(u)-w(v,u)>b\text{ or }\sigma(u)>b\},
\end{align*}
where $\sigma:V\to\{0,\ldots,b+1\}$.
Then it follows from the results of~\cite{AmramEA20} that player
$\mathsf{Eloise}$ wins a node $v$ in the energy parity game
with minimal initial credit $c<b+1$ if 
$(\mathsf{E}^{f^\mathsf{e}_\exists})(v)=c$, that is,
if the solution of the canonical equation system over $f^\mathsf{e}_\exists$
maps $v$ to a value $c$ that is at most $b$.

The \emph{energy $\mu$-calculus}~\cite{AmramEA20} is 
the fixpoint logic that corresponds to energy parity games.
Its formulae
are evaluated over \emph{weighted game structures}
and involve operators
$\Diamond_E\phi$ and $\Box_E\phi$
that are evaluated depending on the energy function 
$\sem{\phi}:V\to\{0,\ldots,b+1\}$ that is obtained by first
evaluating the argument formula $\phi$.
The semantics of the diamond operator then is an energy function that
assigns, to each state $v$, 
the least energy value $c\in\{0,\ldots,b+1\}$ such that 
there is a move from $v$ to some node $u$
such that the credit $c$ suffices to take the move from $v$ to $u$
and retain an energy level of at least $\sem{\phi}(u)$.
Formulae can be translated to equation systems over the finite lattice
$(b+1)^{V}$, where the functions for modal operators are defined according
to their semantics as presented in~\cite{AmramEA20}.
Solving these equation systems then amounts to model checking 
energy $\mu$-calculus formulae over weighted game structures.
\item \emph{Latticed $\mu$-calculi:} 
In latticed $\mu$-calculi~\cite{BrunsGodefroid04}, formulae
are evaluated over complete lattices $L$ rather than the powerset
lattice; for finite lattices $L$, formulae of latticed $\mu$-calculi
hence can be translated to fixpoint equation systems over~$L$, 
so that model
checking reduces to solving equation systems.
An associated latticed variant of games has been introduced 
in~\cite{KupfermanLustig07} and for finite lattices $L$,
solving latticed parity games over $L$ reduces to solving equation systems
over $L$.

\item \emph{The coalgebraic $\mu$-calculus
and coalgebraic parity games:}
The coalgebraic $\mu$-calculus~\cite{CirsteaEA11a}  supports generalized modal branching types
by using \emph{predicate liftings} to interpret formulae over $T$-coalgebras, that is, over structures
whose transition type is specified by an endofunctor $T$ on the
category of sets.
For instance the functors $T=\Pow$, $T=\mathcal{D}$ and
$T=\mathcal{G}$ map sets $X$ to their powerset $\Pow(X)$,
the set of probability distributions $\mathcal{D}(X)=
\{f: X\to[0,\ldots,1]\}$ over~$X$, and to the set of multisets $\mathcal{G}(X)=\{f:X\to \mathbb{N}\}$ over $X$, respectively. The corresponding
$T$-coalgebras then are Kripke frames (for $T=\Pow$), Markov chains (for $T=\mathcal{D}$) and graded transition systems (for $T=\mathcal{G}$),
 respectively.
 Instances of the coalgebraic $\mu$-calculus comprise, e.g.
the two-valued probabilistic $\mu$-calculus~\cite{CirsteaEA11a,LiuEA15}
with modalities $\Diamond_p\phi$ for $p\in[0,\ldots,1]$, expressing 
 `the next state satisfies $\phi$ with probability more than $p$'; 
 the graded $\mu$-calculus~\cite{KupfermanEA02} with
 modalities $\Diamond_g\phi$ for $g\in\mathbb{N}$, expressing 
 `there are more than $\phi$ successor states that satisfy $\phi$';
 or the alternating-time $\mu$-calculus~\cite{AlurEA02} that
 is interpreted over concurrent game frames and uses
 modalities $\langle D\rangle\phi$ for finite $D\subseteq \mathbb{N}$ (encoding
 a \emph{coalition}) that express that
 `coalition $D$ has a joint strategy to enforce $\phi$'.

It has been shown in
previous work~\cite{HausmannSchroder19b} that model checking for
coalgebraic $\mu$-calculi against coalgebras with state space $U$
reduces to solving a canonical fixpoint equation
system over the powerset lattice $\Pow(U)$, where the involved
function interprets modal operators using predicate liftings, as described
in~\cite{CirsteaEA11a,HausmannSchroder19b}.
This canonical equation system can alternatively be seen as
the winning region of $\mathsf{Eloise}$ in \emph{coalgebraic parity games},
a highly general variant of parity games where the game structure is
a coalgebra and nodes are annotated with modalities.
Examples include \emph{two-valued probabilistic parity games} 
and \emph{graded parity games} in which nodes and edges are annotated
with probabilities or grades, respectively. In order to win a
node $v$, player $\mathsf{Eloise}$ then has to have a strategy that 
picks a \emph{set} of moves to nodes that in turn are \emph{all} won by
$\mathsf{Eloise}$, and such that the joint probability (joint
grade) 
of the picked moves is greater than the probability (grade)
that is assigned to $v$. It is known that
solving coalgebraic parity games reduces to
solving fixpoint equation systems~\cite{HausmannSchroder19b}.

Furthermore, the satisfiability problem of the coalgebraic $\mu$-calculus
has been reduced to solving canonical fixpoint equations systems over
lattices
$\Pow(U)$, where~$U$ is the state set of a determinized parity automaton
and where the innermost equation checks for joint one-step satisfiability
of sets of coalgebraic modalities~\cite{HausmannSchroder19a}. 
By interpreting coalgebraic formulae over finite lattices $d^U$ 
rather than over powerset lattices, one obtains the \emph{finite-valued coalgebraic
$\mu$-calculus} (with values $\{0,\ldots,d\}$), which has
the \emph{finite-valued probabilistic $\mu$-calculus} (e.g.~\cite{Mio12})
as an instance.  
Model checking for the finite-valued probabilistic $\mu$-calculus
hence reduces to solving equation systems over the finite lattice
$d^{|U|}$, where $\{0,\ldots,d\}$ encodes a finite set of probabilities.
 
\end{enumerate}

\end{expl}

%
%\begin{definition}[Complement, dual function]
%Given a set $V\subseteq U$, we define the \emph{complement} of
%$V$ in $U$ by $\overline{V}=U\setminus V$, having $\overline{\overline{V}}=V$. Given a monotone
%function $f:(\Pow(U))^k\to\Pow(U)$, we then define the \emph{dual function}
%$\overline{f}$ by putting
%$\overline{f}(U_1,\ldots,U_k)=\overline{f(\overline{U_1},\ldots,
%\overline{U_k})}$ for $(U_1,\ldots,U_k)\in(\Pow (U))^k$.
%\end{definition}
%
%\begin{lemma}[Nested fixpoint duality]\label{fact:fpduality}
%For all $k\in\mathbb{N}$, all $0\leq i\leq k$, all monotone functions $\alpha:(\Pow(U))^k\to\Pow(U)$ and all $(X_{i+1},\ldots,X_k) \in (\Pow(U))^{k-i}$, we have
%\begin{align*}
%\mathsf{E}^\alpha_i(X_{i+1},\ldots,X_k)=\overline{\mathsf{A}^{\overline{\alpha}}_i(\overline{X_{i+1}},\ldots,\overline{X_k})};
%\end{align*}
%in particular, we have $\mathsf{E}^\alpha_k=\overline{\mathsf{A}^{\overline{\alpha}}_k}$.

%\end{lemma}
% \begin{proof}
% This proof is standard and can be found in the
% appendix at the end of this document. \qed
% \end{proof}

\section{Fixpoint Games and History-free Witnesses}

We instantiate the existing notion of fixpoint
games~\cite{Venema08,BaldanEA19}, which characterize solutions
of equation systems, to our setting (that is, to finite lattices), and
then use these games as a technical tool
to establish our crucial notion of history-freeness for
systems of fixpoint equations.

\begin{defn}[Fixpoint games]
  Let $X_i=_{\eta_i} f_i(X_0,\ldots,X_k)$, $0\leq i\leq k$, be a
  system of fixpoint equations.  The associated \emph{fixpoint game}
  is a parity game $(V,E,\Omega)$ with set of nodes
  $V=(B_L\times[k])\cup L^{k+1}$, where nodes from
  $B_L\times[k]$ belong to player $\mathsf{Eloise}$ and nodes
  from $L^{k+1}$ belong to player $\mathsf{Abelard}$.  For nodes
  $(u,i)\in B_L\times[k]$, we put
\begin{align*}
E(u,i)&=\{(U_0,\ldots,U_k)\in L^{k+1}\mid
u\sqsubseteq f_i(U_0,\ldots,U_k)\},
\end{align*} 
and for nodes $(U_0,\ldots,U_k)\in L^{k+1}$, we put 
\begin{equation*}  E(U_0,\ldots,U_k)=\{(u,i)\in B_L\times[k]\mid  u\sqsubseteq U_i\}.
\end{equation*}
The \emph{alternation depth} $\mathsf{ad}(i)$ of an equation $X_i=_{\eta_i}f_i(X_0,\ldots,X_1)$ is defined as $\mathsf{ad}^\mu_i$ if
$\eta_i=\mu$ and as $\mathsf{ad}^\nu_i$ if $\eta_i=\nu$, where $\mathsf{ad}^\mu_i$, $\mathsf{ad}^\nu_i$ are recursively defined by
\begin{align*}
\mathsf{ad}^\mu_i&=\begin{cases}
\mathsf{ad}^\mu_{i-1} & i>0, \eta_{i-1}=\mu\\
\mathsf{ad}^\nu_{i-1}+1 & i>0, \eta_{i-1}=\nu\\
1 & i=0
\end{cases} & 
\mathsf{ad}^\nu_i&=\begin{cases}
\mathsf{ad}^\mu_{i-1}+1 & i>0, \eta_{i-1}=\mu\\
\mathsf{ad}^\nu_{i-1} & i>0, \eta_{i-1}=\nu\\
0 & i=0
\end{cases} 
\end{align*}
for $0\leq i\leq k$.
The priority function $\Omega: V\to[\mathsf{ad}(k)]$ then is defined by
$\Omega(u,i)=\mathsf{ad}(i)$ and $\Omega(U_0,\ldots,U_k)=0$.
\end{defn}

\begin{rem}
In~\cite{BaldanEA19}, an alternative priority function $\Omega': V\to[2k+1]$ with
\begin{align*}
\Omega'(u,i)=\begin{cases}2i & \text{ if }\eta_i=\GFP\\
2i+1& \text{ if }\eta_i=\LFP
\end{cases}
\end{align*}
and $\Omega'(U_0,\ldots,U_k)=0$ is used.  Since $\mathsf{ad}(i)$ is
even if and only if $\eta_i$ is even, and moreover $\mathsf{ad(i)}\le\mathsf{ad}(j)$ for $i\le j$, and $i<j$ whenever $\mathsf{ad(i)}<\mathsf{ad}(j)$, it is easy to see that $\Omega$
and $\Omega'$ in fact assign identical parities to all plays.  In the
following, we will use the more economic parity function $\Omega$ so
that fixpoint games have only $d:=\mathsf{ad}(k)\leq k$ priorities.
\end{rem}
We import the associated characterization theorem \cite[Theorem 4.8]{BaldanEA19}:
\begin{theorem}[\hspace{-0.1pt}\cite{BaldanEA19}]\label{lem:fpgame}
  We have $u\sqsubseteq \sem{X_i}_{f}$ if and only if $\mathsf{Eloise}$ wins
  the node $(u,i)$ in the fixpoint game for the given system~$f$ of
  equations.
\end{theorem}
\begin{rem}
While this shows that parity game solving can 
be used to solve equation systems,
the size of fixpoint games is exponential in 
$|B_L|$, so they do not directly 
yield a quasipolynomial algorithm for solving equation systems.
\end{rem}
Next we define our notion of history-freeness for systems of fixpoint equations.

\begin{defn}[History-free witness]
  A \emph{history-free witness} for $u\sqsubseteq\sem{X_i}_f$ 
  is an even labelled graph $(W,R)$ with labels
  from~$[d]$ such that $W\subseteq B_L\times[d]$, 
  $(u,i)\in W$, and for
  all $(v,p) \in W$, we have $v\sqsubseteq f_p(U_0,\ldots, U_k)$
  where $U_j=\bigsqcup \pi_1[R_{\mathsf{ad}(j)}(v,p)]$ for $0\leq j\leq k$,
  noting that 
  $R_{\mathsf{ad}(j)}(v,p)\subseteq W$ so that 
  $\pi_1[R_{\mathsf{ad}(j)}(v,p)]\subseteq B_L$
  and $U_j\in L$.
\end{defn}
In analogy to history-free strategies for parity games,
history-free witnesses assign tuples $(R_1(v,p),\ldots,R_d(v,p))$ 
of sets $R_j(v,p)\subseteq W$ to
pairs $(v,p)\in W$ without relying on a history of previously
visited pairs.  
We have $|W|\leq (d+1)|B_L|$ and $|R|\leq (d+1)|W|^2$, that is,
the size of history-free witnesses is polynomial in $|B_L|$.  Crucially,
history-free witnesses always exist:
\begin{lemma}\label{lem:witness}
For all $u\in B_L$ and $i\in[k]$, we have
\begin{align*}
u\sqsubseteq \sem{X_i}_f\text{ if and only if there is a history-free witness for }u\sqsubseteq\sem{X_i}_f.
\end{align*}
\end{lemma}
\begin{proof}
  In one direction, we have $u\sqsubseteq\sem{X_i}_f$ so that $\mathsf{Eloise}$ wins the
  node $(u,i)$ in the according fixpoint game by
  Lemma~\ref{lem:fpgame}.  Let $s$ be a corresponding
  \emph{history-free} winning strategy (such strategies always
  exists, see e.g.~\cite{GraedelThomasWilke02}). 
   We inductively construct a
  witness for $u\sqsubseteq\sem{X_i}_f$, starting at
  $(u,i)$. When at $(v,p)\in B_L\times[k]$
  with $s(v,p)=(U_0,\ldots, U_k)$, 
  we put $R_i(v,p)=\bigcup_{j\mid\mathsf{ad}(j)=i}(U_j\times\{ j\})$ for $0\leq i\leq d$
  and hence have $\mathsf{ad}(j)=i$ for all $((v,p),i,(u,j))\in R$. 
   Since~$s$ is a winning strategy,
  the resulting graph $(W,R)$ is a history-free witness for
  $u\sqsubseteq\sem{X_i}_f$ by construction; in
  particular,~$(W,R)$ is even. 
  For the converse direction, the witness for 
  $u\sqsubseteq\sem{X_i}_f$ directly yields a winning
  $\mathsf{Eloise}$-strategy for the node $(u,i)$ in the associated fixpoint
  game. This implies $u\sqsubseteq\sem{X_i}_f$
  by Lemma~\ref{lem:fpgame}.\qed

\end{proof}

\section{Solving Equation Systems using Universal Graphs}

We go on to prove our main result. To this end,
we fix a system $f$ of fixpoint
equations $f_i:L^{k+1}\to L$, $0\leq i\leq k$, and put $n:=|B_L|$ and 
$d:=\mathsf{ad}(k)$ for the remainder of the paper.

\begin{defn}[Universal graphs~\cite{CzerwinskiEA19,ColcombetFijalkow19}]
  Let $G=(W,R)$ and $G'=(W',R')$ be labelled graphs with labels from
  $[d]$. A \emph{homomorphism of labelled graphs}
  from $G$ to $G'$ is a
  function $\Phi:W\to W'$ such that for all $(v,p,w)\in R$, we have
  $(\Phi(v),p,\Phi(w))\in R'$.  An \emph{$(n,d+1)$-universal graph} $S$
  is an even graph with labels from~$[d]$ such that for all
  even graphs $G$ with labels from~$[d]$ and with
  $|G|\leq n$, there is a homomorphism from $G$ to $S$.
\end{defn}
\noindent We fix an
$(n(d+1),(d+1))$-universal graph $S=(Z,K)$,
noting that there are $(n(d+1),(d+1))$-universal graphs (obtained from
universal trees) of size
quasipolynomial in~$n$ and $d$~\cite{CzerwinskiEA19}.  We now combine
the system $f$ with the universal graph $S$ to turn the parity conditions
associated to general systems of fixpoint equations into a safety
condition, associated to a single greatest fixpoint equation.
\begin{defn}[Chained-product fixpoint]\label{defn:timeouts}
We define a function
\begin{align*}
  g\colon\Pow(B_L\times [k]\times Z)&\to\Pow(B_L\times [k]\times Z)\\
  U \qquad& \mapsto \{(v,p,q)\in B_L\times [k]\times Z \mid v\sqsubseteq f_p(P^{U,q}_0,\ldots,P^{U,q}_k)\}
\end{align*} 
% \begin{align*}
% g(U)&=
% \end{align*}
% for $U\in\Pow(B_L\times [k]\times Z)$, 
where
\begin{align*}
  P^{U,q}_i=\bigsqcup\{u\in B_L\mid \exists s\in K_{\mathsf{ad}(i)}(q).\, (u,i,s)\in U\}.
\end{align*} 
We refer to $Y_0=_{\GFP} g(Y_0)$ as the \emph{chained-product fixpoint (equation)} 
of $f$ and $S$.
\end{defn} 
We now show our central result: apart from the annotation with  states from the universal graph, the chained-product
fixpoint $g$ is the solution of the system $f$.
\begin{theorem}\label{thm:main}
For all $u\in B_L$ and
$0\leq i\leq k$, we have 
\begin{align*}
u\sqsubseteq \sem{X_i}_f\text{ if and only if there is }q\in Z
\text{ such that }(u,i,q)\in\sem{Y_0}_g.
\end{align*}
\end{theorem}

\begin{proof}
  For the forward direction, let $u\sqsubseteq \sem{X_i}_f$.  By
  Lemma~\ref{lem:witness}, there is a history-free witness $G=(W,R)$
  for $u\sqsubseteq\sem{X_i}_f$. Since $S$ is a
  $(n(d+1),d+1)$-universal graph and since $G$ is a witness and hence
  an even labelled graph of suitable size $|G|\leq n(d+1)$, there is a
  graph homomorphism $\Phi$ from $G$ to $S$.
  Starting at $(u,i,\Phi(u,i),0)$, we inductively construct a witness for
  containment of $(u,i,\Phi(u,i))$ in $\sem{Y_0}_g$. When at
  $(v_1,p_1,\Phi(v_1,p_1),0)$ with $(v_1,p_1)\in W$, we put
\begin{align*}
R'_0(v_1,p_1,\Phi(v_1,p_1),0)=&\{(v_2,p_2,\Phi(v_2,p_2),0)\in B_L\times[d]\times Z\times[0]\mid\\
&(v_2,p_2)\in R_{\mathsf{ad}(p_2)}(v_1,p_1),\Phi(v_2,p_2)\in K_{\mathsf{ad}(p_2)}(\Phi(v_1,p_1))\,\}
\end{align*} 
and continue the inductive construction with all these 
$(v_2,p_2,\Phi(v_2,p_2),0)$, having $(v_2,p_2)\in W$.
The resulting structure $G'=(W',R')$ indeed is a witness for
containment of $(u,i,q)$ in $\sem{Y_0}_g$: $G'$ is even by
construction. Moreover, we need to show that for 
$(v_1,p_1,\Phi(v_1,p_1),0)\in W'$, we have
$(v_1,p_1,\Phi(v_1,p_1),0)\in g(\pi_1[R'_0(v_1,p_1,\Phi(v_1,p_1),0)])$, i.e. $v_1\sqsubseteq
f_{p_1}(P^{U,\Phi(v_1,p_1)}_0,\ldots,P^{U,\Phi(v_1,p_1)}_k)$
where $U=\pi_1[R'_0(v_1,p_1,\Phi(v_1,p_1),0)]$. Since
$G$ is a witness and $(v_1,p_1)\in W$ by construction of $W'$, we have
$v_1\sqsubseteq f_{p_1}(U_0,\ldots,U_k)$ where
$U_j=\bigsqcup(\pi_j[R_{\mathsf{ad}(i)}(v_1,p_1)])$.
By monotonicity of $f_{p_1}$, it thus suffices to show that
$U_j\sqsubseteq P^{U,\Phi(v_1,p_1)}_j$ for
$0\leq j \leq k$; by definition of $P^{U,\Phi(v_1,p_1)}_j$ this follows if
\begin{align*}
\pi_1[R_{\mathsf{ad}(j)}(v_1,p_1)]\subseteq&\{u\in B_L\mid \exists s\in K_{\mathsf{ad}(j)}(\Phi(v_1,p_1)).
(u,j,s)\in W\},
\end{align*}
where $W=\pi_1[R'_0(v_1,p_1,q_1,0)]$.
So let $w\in B_L$ such that 
$w\in\pi_1[R_{\mathsf{ad}(j)}(v_1,p_1)]$. 
Since $R$ is a witness that is constructed as in the proof
of Lemma~\ref{lem:witness}, we have $i=\mathsf{ad}(i')$ for all 
$((v',p'),i,(w',i'))\in R$. Thus $(w,j)\in R_{\mathsf{ad}(j)}(v_1,p_1)$
for some $j$ such that $\mathsf{ad}(j)=i$, that is,
$((v_1,p_1),\mathsf{ad}(j),(w,j))\in R$, hence
$(\Phi(v_1,p_1),\mathsf{ad}(j), \Phi(w,j))\in K$ because~$\Phi$ is a graph
homomorphism. By definition of $R'_0$ we have
$(w,j,\Phi(w,j),0)\in R'_0(v_1,p_1,\Phi(v_1,p_1),0)$
so that $(w,j,\Phi(w,j))\in \pi_1[R'_0(v_1,p_1,\Phi(v_1,p_1),0)]$.
We are done since $\Phi(w,j)\in K_{\mathsf{ad}(j)}(\Phi(v_1,p_1))$.
% This implies in particular the safety property that every
% edge in $G$ has an according edge in $G'$ (which follows from $\Phi$
% being a graph homomorphism).

For the converse implication, let $(u_0,p_0,q_0)\in \sem{Y_0}_g$ 
for some $q_0\in Z$.
Let $G=(W,R)$ be a history-free witness for this fact.
By Lemma~\ref{lem:fpgame}, it suffices to provide a strategy 
in the fixpoint game for the system $f$ with
which $\mathsf{Eloise}$ wins the node $(u_0,p_0)$.
We inductively construct a \emph{history-dependent} strategy $s$ as follows:
For $i\geq 0$, we abbreviate $U_i=R_0(u_i,p_i,q_i,0)$.
We put $s(u_0,p_0)=(P^{U_0,q_0}_0,\ldots,P^{U_0,q_0}_k)$.
For the inductive step, let 
\begin{align*}
\tau=&(u_0,p_0),(P^{U_0,q_0}_0,\ldots, P^{U_0,q_0}_k),\ldots,(P^{U_{n-1},q_{n-1}}_0,\ldots, P^{U_{n-1},q_{n-1}}_k),(u_n,p_n)
\end{align*}
be a partial play of the fixpoint game that
follows the strategy that has been constructed
so far. Then we have an $R$-path $(u_0,p_0,q_0,0),(u_1,p_1,q_1,0),\ldots,(u_n,p_n,q_n,0)$,
where, for $0\leq i<n$, we have $(q_i,p_{i+1},q_{i+1})\in K$ since
$u_{i+1}\sqsubseteq P^{U_i,q_i}_{p_{i+1}}$ 
by the inductive construction. Put
$s(\tau)=(P^{U_n,q_n}_0,\ldots,P^{U_n,q_n}_k)$. 
Since $G$ is a witness, the strategy uses only moves that are
available to $\mathsf{Eloise}$ (i.e. ones with
$u_n\sqsubseteq f_{p_n}(s(\tau))$). Also, $s$ is a winning strategy
as can be seen by looking at the $K$-paths that are induced by
complete plays $\tau$ that follow $s$, as described (for partial plays)
above. Since $S$ is a universal graph and hence even, every such $K$-path is even and the sequence of priorities in
$\tau$ is just the sequence of priorities of one of these $K$-paths. 
\qed
\end{proof}

\begin{rem}
Since the set $\sem{Y_0}_g$ is the greatest fixpoint of $g$, it
can be computed by simple approximation from above, that is, 
as $g^m(B_L\times[k]\times Z)$ where
$m=|B_L\times [k]\times Z|$. However,
each iteration of the function $g$ may require up to $|Z|$ evaluations
of an equation. In the next section, we will show how this additional
iteration factor in the computation of $\sem{Y_0}_g$ can be avoided.
\end{rem}

\section{A Progress Measure Algorithm}\label{sec:lifting}

We next introduce a lifting algorithm that computes the
set~$\sem{Y_0}_g$ efficiently, following the paradigm of
the progress measure
approach for parity games
(e.g.~\cite{Jurdzinski00,JurdzinskiLazic17}).
Our progress measures will map pairs $(u,i)\in B_L\times[k]$
to nodes in a universal graph that is equipped with
a simulation order, that is, a total order that is suitable for
measuring progress.

\begin{defn}[Simulation order]
For natural numbers $i$, $i'$, we put $i\succeq i'$ if and only if
either $i$ is even and $i=i'$, or both $i$ and $i'$ are odd and
$i\geq i'$.   A total order~$\leq$ on~$Z$ is a \emph{simulation order}
if for all $q,q'\in Z$,
\begin{align*}
q\leq q'\text{  implies that} 
&\text{ for all $0\leq i\leq k$ and $s\in K_{i}(q)$, there are}\\
&\text{ \quad\,\,$i'\succeq i$ and $s'\in K_{i'}(q')$ such that $s\leq s'$}.
\end{align*}
%(informally speaking, $\le$ is a simulation up to increasing odd labels).
% for all sequences $\overline{s}$ of 
%numbers such that there is a $K$-path that starts at $q$ and has
%transitions according to $\overline{s}$, there
%is a $K$-path that starts at $q'$ and has
%transitions according to $\overline{s}$.
\end{defn}

\begin{lemma}\label{lem:order}
There is an $(n(d+1),d+1)$-universal graph $(Z,K)$ of size quasipolynomial in
$n$ and $d$, and
over which a simulation order $\leq$ exists.
\end{lemma}
\begin{proof}[Sketch]
It has been shown \cite[Theorem 2.2]{CzerwinskiEA19} (originally, in
  different terminology,~\cite{JurdzinskiLazic17}) that there are
 \emph{$(l,h)$-universal trees} (a concept similar to, but
slightly more concrete than universal graphs) with set of leaves $T$ 
such that $|T|\leq {2l} {{\log l+h+1}\choose{h}}$.
Leaves in universal trees are identified by
\emph{navigation paths}, that is, sequences of branching directions,
so that the leaves are linearly ordered by the lexicographic order
$\leq$ on navigation paths (which orders leaves from the
left to the right). 
As described in \cite{ColcombetFijalkow19}, one can obtain a universal graph
$(T,K)$ over $T$ in which transitions $(q,i,q')\in K$ for odd $i$ (the crucial case)
move to the left, that is, $q'$ is a leaf that is to the left of $q$ in 
the universal tree (so that $q'<q$), ensuring universality. 
As it turns out, the lexicographic
ordering on $T$ is a simulation order.
Adapting this construction to our setting, we put $l=n(d+1)$ 
and $h=d+1$ and obtain a
$(n(d+1),d+1)$-universal graph (along with a simulation order $\leq$) of
size at most 
${2n(d+1)} {{\log (n(d+1))+d+2}\choose{d+1}}$
which is quasipolynomial in $n$ and $d$.
\qed
\end{proof}
We fix an $(n(d+1),d+1)$-universal graph $(Z,K)$ and a simulation
order~$\leq$ on~$Z$ for the remainder of the paper (these exist by the
above lemma).

\begin{defn}[Progress measure, lifting function]\label{defn:pm}
We let $q_\mathsf{min}\in Z$ denote the least node w.r.t. $\leq$
and fix a distinguished top element $\star\notin Z$,
and extend $\geq$ to $Z\cup\{\star\}$ by putting 
$\star\geq q$ for all $q\in Z$.
  A \emph{measure} is a map
  $\mu\colon B_L\times[k]\to Z\cup\{\star\}$, i.e.\ assigns nodes in
  the universal graph or $\star$ to pairs
  $(v,p)\in B_L\times[k]$. A measure~$\mu$ is a \emph{progress
    measure} if whenever $\mu(v,p)\neq\star$, then
  $v\sqsubseteq f_p(U_0^{\mu,q},\ldots,U_k^{\mu,q})$ where $q=\mu(v,p)$ and
\begin{align*}
U_i^{\mu,q}=
\bigsqcup
\{u\in B_L\mid \exists s\in K_{\mathsf{ad}(i)}(q).\,\mu(u,i)\leq s\}.
\end{align*} 
We define a function 
$\mathsf{Lift}:(B_L\times[k]\to Z\cup\{\star\})\to (B_L\times[k]\to Z\cup\{\star\})$ on measures by 
\begin{align*}
(\mathsf{Lift}(\mu))(v,p)=
\min \{q\in Z\mid v\sqsubseteq f_p(U_0^{\mu,q},\ldots,U_k^{\mu,q}) \} 
\end{align*}
where $\min(Z')$
denotes the least element of $Z'$ w.r.t. $\leq$, for
$\emptyset\neq Z'\subseteq Z$; also we put $\min(\emptyset)=\star$.
\end{defn}
The lifting algorithm then starts with the least measure 
$\mathsf{m}_{\min}$ that maps all pairs
$(v,p)\in B_L\times [k]$ to the minimal node 
(i.e.~$\mathsf{m}_{\min}(v,p)=q_{\min}$) and
repeatedly updates the current measure
using $\mathsf{Lift}$ until the measure
stabilizes.
\mysubsec{Lifting algorithm} 
\begin{enumerate}
\item Initialize: Put
$\mu:=\mathsf{m}_{\min}$.
\item If $\mathsf{Lift}(\mu)\neq \mu$, then
put $\mu:=\mathsf{Lift}(\mu)$ and go to 2. Otherwise go to 3.
\item Return the set $\mathbb{E}=\{(v,p)\in B_L\times [k]\mid \mu(v,p)\neq \star\}$.
\end{enumerate}

\begin{lemma}[Correctness]\label{lem:liftalg}
For all $v\in B_L$ and $0\leq p\leq k$,
we have
\begin{align*}
(v,p)\in\mathbb{E} \text{ if and only if $v\in\sem{X_p}_f$}.
\end{align*}
\end{lemma}
\begin{proof}[Sketch]
Let $\mu$ denote the progress measure that the algorithm computes.
For one direction of the proof, let $(v,p)\in\mathbb{E}$. 
By Lemma~\ref{lem:witness} it suffices to construct a witness
for $v\in\sem{X_p}_f$. We extract such a witness $(\mathbb{E},R)$
from the progress measure $\mu$, relying on 
the properties of the simulation order $\leq$ that is used to measure
the progress of $\mu$
to ensure that
any infinite sequence of measures that $\mu$ assigns to some $R$-path
induces an infinite (and hence even) path in the employed
universal graph. This shows that $(\mathbb{E},R)$ indeed is an even graph
and hence a witness. For the converse direction, let
$v\in\sem{X_p}_f$ so that there is, by Theorem~\ref{thm:main},
some $q\in Z$ such that $(v,p,q)\in\sem{Y_0}_g$.
For $(u,i)$ such that
there is $q'\in Z$  such that $(u,i,q')\in\sem{Y_0}_g$, let
$q_{(u,i)}\in Z$ denote the minimal  such node w.r.t.~$\leq$.
It now suffices that
  $\mu(u,i)\leq q_{(u,i)}$ for all such $(u,i)$, which is shown by
  induction on the number of iterations of the lifting algorithm.\qed
\end{proof}

\begin{corollary}\label{cor:qp}
Solutions of systems of fixpoint equations can be computed with
quasipolynomially many evaluations of equations.
\end{corollary}
\begin{proof}
Given an $(n(d+1),d+1)$-universal graph $(Z,K)$ and a 
simulation order on $Z$,
the lifting algorithm terminates and returns the solution
of $f$ after at most $n(d+1)\cdot |Z|$ many iterations. This is 
the case since each iteration
(except the final iteration) increases the measure for at least one of
the $n(d+1)$ nodes
and the measure of each node can be increased at most $|Z|$ times.
Using the universal graph and the simulation order
from the proof of Lemma~\ref{lem:order},
we have $|Z|\leq {2n(d+1)} {{\log (n(d+1))+d+2}\choose{d+1}}$
so that the algorithm terminates after at most 
$2(n(d+1))^2 {{\log (n(d+1))+d+2}\choose{d+1}}\in 
\mathcal{O}((n(d+1))^{\log (d+1)})$
iterations of the function
$\mathsf{Lift}$.
Each iteration can be implemented to run with at most
$n(d+1)$ evaluations of an equation.
\qed
\end{proof}

\begin{corollary}\label{cor:fewpriopoly}
The number of function calls required
for the solution of systems of fixpoint equations with
$d\leq\log n$ is bounded by a polynomial in $n$ and
$d$.
\end{corollary}
\begin{proof}
Following the insight of Theorem 2.8 in \cite{CaludeEA17}, 
Theorem 2.2. in \cite{CzerwinskiEA19} implies that
if $d< \log n$, then there is an $(n(d+1),d+1)$-universal tree of size polynomial in $n$ and $d$. In the same way as in the proof of 
Lemma~\ref{lem:order}, one obtains a universal graph
of polynomial size and a simulation order on it. \qed
\end{proof}

\begin{expl}
Applying Corollary~\ref{cor:qp} and Corollary~\ref{cor:fewpriopoly} to 
Example~\ref{exmp:eqs}, we obtain the following results:
\begin{enumerate}
\item The model checking problems for the energy $\mu$-calculus 
and finite latticed $\mu$-calculi are in $\mathsf{QP}$.
For energy parity games with sufficient upper bound $b$ on energy level accumulations,
we obtain a progress measure algorithm that terminates
after a number of iterations that is quasipolynomial in
$b$.
\item Under mild assumptions on the modalities (see~\cite{HausmannSchroder19b}), the model checking problem 
for the coalgebraic $\mu$-calculus is in $\mathsf{QP}$; in particular,
this yields $\mathsf{QP}$ model checking algorithms for the
graded $\mu$-calculus and the two-valued probabilistic $\mu$-calculus
(equivalently: $\mathsf{QP}$ progress measure algorithms for solving
graded and two-valued probabilistic parity games).
\item Under mild assumptions on the modalities 
(see~\cite{HausmannSchroder19a}),
  we obtain a novel upper bound
  $2^{\mathcal{O}({nd\log n})}$ for the satisfiability problems of
  coalgebraic $\mu$-calculi, in particular including the monotone $\mu$-calculus,
  the alternating-time $\mu$-calculus, the graded $\mu$-calculus and
  the (two-valued) probabilistic $\mu$-calculus, even when the latter
  two are extended with (monotone) polynomial inequalities. This
  improves on the best previous bounds in all cases.
\end{enumerate}
\end{expl}

\section{Conclusion}
We have shown how to use universal graphs to compute solutions of
systems of fixpoint equations $X_i=\eta_i.\, f_i(X_0,\ldots,X_k)$ (with
the~$\eta_i$ marking least or greatest fixpoints) that use functions
$f_i:L^{k+1}\to L$ (over a finite lattice~$L$ with basis $B_L$) 
and involve up to
$k+1$-fold nesting of fixpoints. Our progress measure
algorithm needs quasipolynomially
many evaluations of equations, and runs in time
$\mathcal{O}(q\cdot t(f))$, where
$q$ is a quasipolynomial in $|B_L|$ and the alternation depth
of the equation system, and
where~$t(f)$ is an upper bound on the time it takes to compute $f_i$
for all $i$.  % In particular, if~$\alpha$ can be computed
% in quasipolynomial time, then its $k$-nested fixpoint can be
% computed in quasipolynomial time as well.

As a consequence of our results, the upper time bounds for
the evaluation of various general parity conditions improve.
Example domains beyond solving parity games to which our algorithm
can be instantiated comprise
model checking for latticed $\mu$-calculi and solving
latticed parity games\cite{BrunsGodefroid04,KupfermanLustig07},
solving energy parity games and model checking for
the energy $\mu$-calculus~\cite{{ChatterjeeDoyen12,AmramEA20}},
and model checking
and satisfiability checking for the coalgebraic $\mu$-calculus~\cite{CirsteaEA11a}. 
The resulting model checking algorithms for latticed $\mu$-calculi 
and the energy $\mu$-calculus run in time quasipolynomial in the 
provided basis of the respective lattice.
In terms of concrete instances of the coalgebraic $\mu$-calculus, we
obtain, e.g., quasipolynomial-time model checking for the graded~\cite{KupfermanEA02} and the probabilistic
$\mu$-calculus~\cite{CirsteaEA11a,LiuEA15} as new results
(corresponding results for, e.g., the alternating-time
$\mu$-calculus~\cite{AlurEA02} and the monotone
$\mu$-calculus~\cite{EnqvistEA15} follow as well but have already been
obtained in our previous work~\cite{HausmannSchroder19b}), as well as
improved upper bounds for satisfiability checking in the graded
$\mu$-calculus, the probabilistic $\mu$-calculus, the monotone
$\mu$-calculus, and the alternating-time $\mu$-calculus. We foresee
further applications, e.g.\ in the computation of fair bisimulations
and fair equivalence~\cite{HenzingerRajamani00,KupfermanEA03} beyond
relational systems, e.g.\ for probabilistic systems.

As in the case of parity games, a natural open question that remains
is whether solutions of fixpoint equations 
can be computed in polynomial time (which
would of course imply that parity games can be solved in polynomial
time). A more immediate perspective for further investigation is to
generalize the recent quasipolynomial variant~\cite{Parys19} of
Zielonka's algorithm~\cite{Zielonka98} for solving parity games to
solving systems of fixpoint equations, with a view to improving efficiency 
in practice.

% It has also been shown that the annotation of parity games with histories/timeouts has to be done according to universal trees~\cite{CzerwinskiEA19} in order to result in an algorithm that correctly solves the games. Exponential timeouts and their updating procedure then correspond to universal trees of exponential sizes while quasipolynomial histories correspond to universal trees of quasipolynomial sizes and with a rather involved ordering of nodes that is defined according to the updating procedure for quasipolynomial histories from Calude et al.
% In this sense, universal trees can be seen to \emph{stand behind parity games} and their sizes determines the sizes of timeouts/histories.
% We conjecture that universal trees also stand behind the computation of nested fixpoints.
% A quasipolynomial lower bound on the size of universal trees for parity games has been shown and we also conjecture that this lower bound
% extends to the computation of nested fixpoints by means of annotation with timeouts/histories.

%% Acknowledgments
%\begin{acks}                            %% acks environment is 
%We would like to thank Barbara K\"onig for bringing
%fixpoint games to our attention.
%\end{acks}

%% Bibliography
\bibliographystyle{splncs04}
\bibliography{coalgml}

%\end{document}

%% Appendix
\appendix
\section{Appendix}

\subsection{Omitted Proofs}

\textbf{Full proof of Lemma~\ref{lem:order}.}
\begin{proof}
It has been shown in 
\cite{CzerwinskiEA19}, Theorem 2.2. 
that there are $(l,h)$-universal trees with set of leaves $T$ 
such that $|T|\leq {2l} {{\log l+h+1}\choose{h}}$; following
 \cite{ColcombetFijalkow19}, we show how to obtain a universal graph
 $(Z,K)$ of the same size. 
The leaves in universal trees can be identified by \emph{navigation paths},
that is, tuples $(m_k,m_{k-2},\ldots,m_1)$
of natural numbers that describe the position of the leaf within the tree.
We put $Z=T$ and define the relation $K$ by putting
\begin{align*}
K_i(m_k,m_{k-2},\ldots,m_1)=
\begin{cases}
\{(m_k,m_{k-2},\ldots,m_i-1)\mid m_i>0\} & \text{if } i \text{ is odd}\\
\{(m_k,m_{k-2},\ldots,m_{i-1})\} & \text{if } i \text{ is even}
\end{cases}
\end{align*}
for $0\leq i\leq k$, $(m_k,m_{k-2},\ldots,m_i)\in Z$.
Here we use $(m_k,m_{k-2},\ldots,m_{p})$ (for odd $p\geq 1$) 
to denote the rightmost
leaf that is reachable from position $(m_k,m_{k-2},\ldots,m_{p})$
in the universal tree. By $(l,h)$-universality of the used universal tree,
$(K,Z)$ is a $(l,h)$-universal graph.
Let $\leq$ be the lexicographic order on the leaves of the universal
tree; then $\leq$ orders the leaves linearly from left to right.
We show that $\leq$ is a simulation order on $(Z,K)$.
Let $q=(m_k,m_{k-2},\ldots,m_1)\in Z$ and 
$q'=(m'_k,m'_{k-2},\ldots,m'_1)\in Z$ such that $q\leq q'$.
Furthermore let $0\leq i\leq k$ and 
$s\in K_i(q)$. We have
to show that there are $i'\succeq i$ and $s'\in K_{i'}(q')$
such that $s\leq s'$. If $q=q'$, then we are done. 
If $q\neq q'$ and
there is $s'\in K_{i}(q')$, then $s'=(m'_k,m'_{k-2},\ldots,m'_i-1)$
(if $i$ is odd) or $s'=(m'_k,m'_{k-2},\ldots,m'_{i-1})$ (if $i$ is even).
In the first case, we have $s=(m_k,m_{k-2},\ldots,m_i-1)$ and in the
second case, we have $s=(m_k,m_{k-2},\ldots,m_{i-1})$. In both cases
we have $i\succeq i$ and
$s\leq s'$ follows from $q\leq q'$.
If $q\neq q'$ and $K_{i}(q')=\emptyset$,
then $i$ is odd (and $m'_i=0$) and
there is an odd number $i'>i$ such that $m_{i'}<m'_{i'}$,
since $q< q'$.
Since $m'_{i'}>m_{i'}\geq 0$, there is $s'=(m_k,m_{k-2},\ldots,m_{i'}-1)\in K_{i'}(q')\neq \emptyset$. It remains to show
$i'\succeq i$ and $s=(m_k,m_{k-2},\ldots,m_{i-1})\leq
(m_k,m_{k-2},\ldots,m_{i'}-1)=s'$. The former
follows from 
$i'>i$ and the fact
that both $i$ and $i'$ are odd, and the latter follows from $i'>i$ together with 
$m_{i'}<m'_{i'}$ by the definition of lexicographic ordering.

 Putting
 $l=n(k+1)$ and $h=k+1$, we hence obtain an $(n(k+1),k+1)$-universal graph
 of (quasipolynomial) size at most ${2n(k+1)} {{\log (n(k+1))+k+2}\choose{k+1}}$ and a simulation order $\geq$ on this graph.\qed
\end{proof}
\textbf{Full proof of Lemma~\ref{lem:liftalg}.}
\begin{proof}
For one direction of the proof, we let 
$\mu$ be the measure that is obtained at the end of the execution of
the lifting algorithm. Then we have $\mathsf{Lift}(\mu)= \mu$ so that
$\mu$ is a progress measure.
Let $(v,p)\in\mathbb{E}$.
By Lemma~\ref{lem:witness} it suffices to exhibit a
history-free witness $(W,R)$ for $v\in\sem{X_p}_f$.
We put $W=\mathbb{E}$, having $(v,p)\in W$,
and define the relation $R\subseteq W\times[d]\times W$ 
by 
\begin{align*}
  R_{i'}(u,i)=
  \{(w,i')\in W\mid \exists s\in K_{\mathsf{ad}(i')}(\mu(u,i)).\,\mu(w,i')\leq  s) \}
\end{align*}
   for $(u,i)\in W$ and $0\leq i'\leq d$, having
$\star\neq\mu(u,i)\in Z$ by definition of $W$.  To show that $(W,R)$
is a history-free witness, we first show that for all $(v,p)\in W$, we
have $v\sqsubseteq f_p(U_0,\ldots,U_k)$ where
\begin{align*}
U_i=\bigsqcup\pi_1[R_{\mathsf{ad}(i)}(v,p)]
\end{align*} for $1\leq i\leq k$.
Since $\mu$ is a progress measure, we have 
\begin{align*}
v\sqsubseteq f_p(U^{\mu,\mu(v,p)}_0,\ldots,U^{\mu,\mu(v,p)}_k).
\end{align*}
% By monotonicity of $f_p$, it suffices to show
% $U_j^{\mu(v,p)}\sqsubseteq U_j$ for $1\leq j\leq k$, which is the case 
Our claim $v\sqsubseteq f_p(U_0,\ldots,U_k)$ follows since 
we have, for all $1\leq i\leq k$, that by the definition 
of~$R_{\mathsf{ad}(i)}$,
\begin{align*}
U_i=\bigsqcup\pi_1[R_{\mathsf{ad}(i)}(v,p)]&=\bigsqcup\pi_1[\{(w,i)\in W\mid 
\exists s\in K_{\mathsf{ad}(i)}(\mu(v,p)).\,\mu(w,i)\leq s \}]\\
&=\bigsqcup\{w\in B_L\mid \exists s\in K_{\mathsf{ad}(i)}(\mu(v,p)).\,\mu(w,i)\leq s\}\\
&=U_i^{\mu,\mu(v,p)}.
\end{align*}
It remains to show that $(W,R)$ is an even graph. To see this, let
\begin{align*}
\pi=(u_0,p_0),p_1,(u_1,p_1),p_2,(u_2,p_2)\ldots
\end{align*}
be an infinite path in $(W,R)$, that is, let
$((u_i,p_i),p_{i+1},(u_{i+1},p_{i+1}))\in R$ for all
$i\geq 0$.  We have
to show that $\pi$ is even. % By definition of $R$, we have some
%$i_{j+1}$ such that
%$i_{j+1}\succeq p_{j+1}$ for all $j\geq 0$, and there is
%$s\in K_{i_{j+1}}(\mu(u_j,p_j))$ such that
%$\mu(u_{j+1},p_{j+1})\leq s$.
% that is, there is 
%an infinite sequence
%\begin{align*}
%\mu(u_{0},p_{0})\geq_{p_1}\mu(u_1,p_1)\geq_{p_2}\mu(u_2,p_2)\geq_{p_3}\ldots
%\end{align*}
We construct an infinite path
\begin{equation*}
\tau=s_0,q_1,s_1,q_2,s_2\ldots
\end{equation*}
in the universal graph $(Z,K)$ as follows: We put
$s_0=\mu(u_0,p_0)$, and define
$q_{i+1}$ and $s_{i+1}$ inductively from $s_i$ for all $i\geq 0$,
ensuring the invariant
\begin{equation}\label{eq:inv}
\mu (u_{i},p_{i})\leq s_i.
\end{equation}
First we note that~\eqref{eq:inv} holds trivially at $i=0$.  Next we
define $q_{i+1}$ and~$s_{i+1}$ for $i\geq 0$, assuming
that~\eqref{eq:inv} holds at~$i$.  By definition of $R$, it follows
from $((u_{i},p_{i}),p_{i+1},(u_{i+1},p_{i+1}))\in R$ that
there is $s\in K_{\mathsf{ad}({p_{i+1}})}(\mu(u_{i},p_{i}))$ such that
$\mu(u_{i+1},p_{i+1})\leq s$.  By the assumption on~$\geq$,
$s\in K_{\mathsf{ad}(p_{i+1})}(\mu(u_{i},p_{i}))$ and the invariant
$s_{i}\geq \mu(u_{i},p_{i})$ imply that there is some $p$
and some $s'\in K_{p}(s_{i})$ such that 
$p\succeq \mathsf{ad}(p_{i+1})$ and $s'\geq s$. Put $q_{i+1}:=p$ and
$s_{i+1}:=s'$; then $s_{i+1}\in K_{q_{i+1}}(s_{i})$ and
$\mu (u_{i+1},p_{i+1})\leq s\leq s_{i+1}$, so the
invariant~\eqref{eq:inv} holds at $i+1$. For later
use, we point out that we
have $q_{i}\succeq \mathsf{ad}(p_{i})$ for all $i> 0$ by construction of $\tau$.

The path $\tau$, being an infinite path in the universal graph
$(Z,K)$, is even.  Since we have $q_{i}\succeq \mathsf{ad}(p_{i})$ for all $i> 0$,
$p_{i}$ is even if and only if $q_{i}$ is even (and then 
$q_i=\mathsf{ad}(p_i)$),
and $p_{i}$ is odd if and only if $q_{i}$ is odd (and then
$q_{i}\geq \mathsf{ad}(p_{i})$). 
 Let $q$ be the highest priority that occurs
infinitely often in $\tau$.  Then $q$ is even and there are infinitely 
many positions
$i$ such that $q_i = q = \mathsf{ad}(p_i)$
and there is some position $i'$ such that for all
positions $i''>i'$, $\mathsf{ad}(p_{i''})\leq q_{i''} \leq q$.
Hence $\pi$ is an even path, as required.

For the converse direction, let $u\in\sem{X_i}_f$.  We note that the
lifting algorithm computes a progress measure as the least fixpoint of
the function $\mathsf{Lift}$ in the lattice of measures, pointwise
ordered by the order $\leq$ on nodes of the universal
graph.  Let $\mu$ be the
measure that is obtained at the end of the execution of the lifting
algorithm. Then we have $\mathsf{Lift}(\mu)=\mu$ and
$\mu=\mathsf{Lift}^n(\mathsf{m}_{\min})$ for some natural number $n$
(recall that $\mathsf{m}_{\min}$ denotes the minimal measure). 
We have to show that
$(u,i)\in\mathbb{E}$, i.e.\ $\mu(u,i)\neq \star$.  By
Theorem~\ref{thm:main}, there is $q\in Z$ such that
$(u,i,q)\in \sem{Y_0}_g$.  For all $(v,p)$ such that there is
$q'\in Z$ such that $(v,p,q')\in\sem{Y_0}_g$, we let $q_{(v,p)}\in Z$
denote the least (w.r.t. $\leq$) node of the universal graph with this
property.  It suffices to show that in this case,
\begin{align*}
(\mathsf{Lift}^n(\mathsf{m}_{\min}))(v,p)\leq q_{(v,p)}.
\end{align*}
We proceed by induction over $n$. If $n=0$, then
$\mathsf{Lift}^n(\mathsf{m}_{\min})=\mathsf{m}_{\min}$ and
$\mathsf{m}_{\min}(v,p)=q_\mathsf{\min}\leq q_{(v,p)}$.  If $n>0$,
then we have to show the inequality in
\begin{align*}
(\mathsf{Lift}^n(\mathsf{m}_{\min}))(v,p)
&= \mathsf{Lift}(\mathsf{Lift}^{n-1}(\mathsf{m}_{\min}))(v,p)\\
&= \min\{q\in Z\mid v\sqsubseteq f_p(U_0^{\mu_{n-1},q},\ldots, 
U_k^{\mu_{n-1},q})\}\leq q_{(v,p)}
\end{align*} 
where $\mu_{n-1}=\mathsf{Lift}^{n-1}(\mathsf{m}_{\min})$.  This
follows once we show that
\begin{align*}
q_{(v,p)}\in\{q\in Z\mid v\sqsubseteq 
f_p(U_0^{\mu_{n-1},q},\ldots, U_k^{\mu_{n-1},q})\},
\end{align*}
i.e. 
\begin{equation}\label{eq:qvp-goal}
v\sqsubseteq 
f_p(U_0^{\mu_{n-1},q_{(v,p)}},\ldots, U_k^{\mu_{n-1},q_{(v,p)}}).
\end{equation}
Now since $\sem{Y_0}_g$ is a fixpoint of $g$, we have
\begin{align*}
\sem{Y_0}_g=g(\sem{Y_0}_g)=\{(v,p,q)\in B_L\times [k]\times Z\mid v\sqsubseteq f_p(P_0^{\sem{Y_0}_g,q},\ldots,P_k^{\sem{Y_0}_g,q})\}.
\end{align*}
So given that $(v,p,q_{(v,p)})\in \sem{Y_0}_g$, we have
\begin{align*}
v\sqsubseteq f_p(P_0^{\sem{Y_0}_g,q_{(v,p)}},\ldots,P_k^{\sem{Y_0}_g,q_{(v,p)}}),
\end{align*}
which implies our goal~\eqref{eq:qvp-goal} by monotonicity of $f_p$
once we show that
$P_i^{\sem{Y_0}_g,q_{(v,p)}}\sqsubseteq U_i^{\mu_{n-1},q_{(v,p)}}$ for
$0\leq i\leq k$.  Both sides of this inequality are defined as joins
of sets; it suffices to show the set inclusion between these sets. So
let $u'\in B_L$ such that there is $q'\in K_{\mathsf{ad}(i)}(q_{(v,p)})$ such that
$(u',i,q')\in \sem{Y_0}_g$ ($P_i^{\sem{Y_0}_g,q_{(v,p)}}$ is the join
of all such~$u'$).  We claim that $u'$ is in the set
\begin{align*}
  \{u''\in B_L\mid \exists s\in K_{\mathsf{ad}(i)}(q_{(v,p)}).\,(\mathsf{Lift}^{n-1}(\mathsf{m}_{\min}))(u'',i)\leq s\}
\end{align*}
(whose join is $U_i^{\mu_{n-1},q_{(v,p)}}$
since $\mu_{n-1}=\mathsf{Lift}^{n-1}(\mathsf{m}_{\min})$). 
Indeed, pick $s=q'$ so that 
$s\in K_{\mathsf{ad}(i)}(q_{(v,p)})$.  It remains to show that
$(\mathsf{Lift}^{n-1}(\mathsf{m}_{\min}))(u',i)\leq s$.  By the
inductive hypothesis,
$(\mathsf{Lift}^{n-1}(\mathsf{m}_{\min}))(u',i)\leq q_{(u',i)}$, so it
suffices to show $q_{(u',i)}\leq s$. This follows from the fact that
$s=q'$ and $(u',i,q')\in \sem{Y_0}_g$, since $q_{(u',i)}$ is by
definition the least node (w.r.t. $\leq$) such that
$(u',i,q_{(u',i)})\in \sem{Y_0}_g$.  \qed
\end{proof}

\subsection{Details on Applications to Coalgebraic $\mu$-Calculi}\label{sect:coalg}

\begin{expl}\label{expl:prob-games}
We consider probabilistic parity games, which make use of
systems of fixpoint equations that deviate considerably from
(and apparently do not reduce easily to) the ones for standard parity games.
Probabilistic parity games are parity games
in which both moves and nodes are annotated with probabilities
(these games are
not to be confused with the $2\frac{1}{2}$-player \emph{stochastic parity games}
that are considered in~\cite{ChatterjeeJurdzinskiHenzinger03,HahnScheweTurriniZhang16}). 
They arise naturally as model checking games for the (two-valued)
probabilistic $\mu$-calculus (see Example~\ref{expl:logics}.\ref{expl:probmu}); we postpone a more formal and detailled treatment to Section~\ref{sect:coalg} below,
where we discuss the more general 
coalgebraic $\mu$-calculus (covering e.g. probabilistic, graded and
the alternating-time $\mu$-calculi as instances)
and its model checking problem (corresponding to solving e.g.
probabilistic, graded and alternating-time parity games).

A \emph{probabilistic parity game} $(V,D,\Omega,\sigma)$
consists of a set
$V$ of nodes, a set of probabilistic moves, given by a function $D$
which assigns probability distributions $D(v)$ over $V$
(with $\Sigma_{w\in V}(D(v))(w)=1$)
to nodes $v\in V$, a priority function $\Omega:V\to\mathbb{N}$
and a \emph{probability assignment} $\sigma:V\to[0,1]$.
The intuition of $(D(v))(w)=p$ is that the move from $v$ node to node $w$ has probability $p$.
A \emph{play} $\rho=v_0,v_1,\ldots$ in a probabilistic 
parity game is a sequence
of nodes such that for all $i\geq 0$, we have $(D(v_i))(v_{i+1})>0$ and the winning
conditions on plays are the same as in standard parity games.
 A \emph{(history-free) strategy} is a
partial function $s:V\rightharpoonup \Pow(V)$ such that
for all $v\in\mathsf{dom}(s)$, $\Sigma_{w\in s(v)}(D(v))(w)>\sigma(v)$,
that is, strategies pick \emph{sets} of moves whose
joint probability is larger than the probability assignment of 
the respective node. Crucially, strategies in probabilistic parity games involve branching, unlike strategies in standard parity games
which pick single moves for each node on which they are defined.
An $s$-play is a play $\rho=v_0,v_1,\ldots$
such that for all $i\geq 0$, $v_i\in\mathsf{dom}(s)$ and
$v_{i+1}\in s(v_i)$.
 Player
$\mathsf{Eloise}$ wins a node $v$ if there is a strategy $s$
such that $\mathsf{Eloise}$ wins all $s$-plays that start at $v$.
We point out that the (somewhat concealed) two-player nature of probabilistic parity games manifests in the fact that 
$\mathsf{Eloise}$ has to pick, in each turn, \emph{some} suitable set of moves, whereupon
$\mathsf{Abelard}$ can challenge \emph{any} of these moves
(that is, \emph{all} $s$-plays need to be even in order for $s$
to be a winning strategy for $\mathsf{Eloise}$).
Player $\mathsf{Eloise}$ hence wins a node $v$
if and only if
there is a set $W\subseteq V$ containing $v$ and
a graph $G=(W,R\subseteq W\times W)$ such that
\begin{itemize}
\item[--] each $R$-edge has $D$-probability
greater than 0,
\item[--] for each $w\in W$, the $R$-successors of $w$ have a joint
$D$-probability of more than $\sigma(w)$ and
\item[--] for each
infinite $R$-path that starts at $v$, the highest priority that is visited
infinitely often by the path is even.
\end{itemize}
Consider, for instance, the probabilistic parity game depicted below with
$V=\{0,1,2,3\}$, $\Omega(i)=i$ for $i\in V$ and, e.g.
$(D(0))(1)=0.2$, $(D(3))(1)=1$ and $(D(3))(3)=0$;
let us fix the probability assigment $\sigma$
by putting $\sigma(0)=0.7$, $\sigma(1)=0.3$
$\sigma(2)=0.1$ and $\sigma(3)=0$.
\begin{center}
\tikzset{every state/.style={minimum size=15pt}}
  \begin{tikzpicture}[
    % Default arrow tip
    %-&gt;,&gt;=stealth',shorten &gt;=1pt,
		auto,
    % Default node distance
    node distance=1.2cm,
    % Edge stroke thickness: semithick, thick, thin
    semithick
    ]
     \node[state] (0) {$0$};
     \node (yo) [below of=0] {};
     \node[state] (1) [left of=yo] {$1$};
     \node[state] (2) [right of=yo] {$2$};
     \node[state] (3) [below of=yo] {$3$};
     \path[->] (0) edge [loop right] node [right] {$0.5$} (0);
     \path[->] (0) edge [bend right=30] node [pos=0.3,left] {$0.2\,\,$} (1);
     \path[->] (0) edge node [pos=0.3,right] {$\,0.3$} (2);
     \path[->] (1) edge [loop left] node [left] {$0.8$} (1);
     \path[->] (1) edge [bend right=30] node [pos=0.3,right] {$\,\,0.2$\;\;\;\;\;} (0);
     \path[->] (2) edge [loop right] node [right] {$0.4$} (2);
     \path[->] (2) edge node [pos=0.6,right] {$\,0.6$} (3);
     \path[->] (3) edge node [pos=0.6,right] {$\,1$} (1);
     
  \end{tikzpicture}
\end{center}
To win e.g. the node $0$, $\mathsf{Eloise}$ has to have
a strategy that selects a set of nodes that have a joint probability
(of being reached from $0$ in one step) greater than $\sigma(0)=0.7$ and that are in turn all won by the strategy.
In this example, player $\mathsf{Eloise}$ wins the nodes $0$ and 
$2$ with the strategy $s$ defined by $s(0)=\{0,2\}$
and $s(2)=\{2\}$: this function indeed is a valid strategy since
it uses only moves with nonzero probabilities and also
respects the probability assignment $\sigma$ as we have
$\Sigma_{j\in s(0)} (D(0))(j)=0.5+0.3 > \sigma(0)=0.7$
and
$\Sigma_{j\in s(2)} (D(2))(j)=0.4 > \sigma(2)=0.1$.
Also, every $s$-play that starts at node $0$
is of the form $0^\omega$ or $0^*2^\omega$ 
and hence even and
every $s$-play that starts at node $2$ is of
the form $2^\omega$ and hence even.
On the other hand, there is no strategy with
which $\mathsf{Eloise}$ can win the nodes $1$ or $3$
since for any strategy $t$ with $1\in\mathsf{dom}(t)$,
we have $\Sigma_{w\in t(1)}(D(1))(w)>\sigma(1)=0.3$ and hence
$1\in t(1)$; but then there is an odd $t$-play of the form
$1^\omega$ so that $t$ is not a winning strategy
for $\mathsf{Eloise}$. Also, for any candidate winning strategy $u$ with
$3\in\mathsf{dom}(u)$, we have $u(3)=\{1\}$ and hence
$1\in\mathsf{dom}(u)$ which shows that $u$ is not a winning
strategy for $\mathsf{Eloise}$. Hence we have
$\mathsf{win}_\exists=\{0,2\}$ and $\mathsf{win}_\forall=\{1,3\}$.

The winning regions in probabilistic parity games are
again just nested fixpoints, where the functions however
deviate significantly from the functions for standard parity games.
Player $\mathsf{Eloise}$ has to pick sets of moves now, so 
it does not suffice to consider existential or universal branching,
like in the function $f_\exists$ for standard parity games.
We define
$f_{\exists p}:\Pow(V)^{k+1}\to\Pow(V)$, for $(V_0,\ldots,V_k)\in\Pow(V)^{k+1}$, by
putting
\begin{align*}
f_{\exists p}(V_0,\ldots,V_k)=&
\{v\in V\mid \exists 0\leq i\leq k.\,\Omega(v)=
i, \\
&\qquad\qquad\quad \Sigma_{w\in V_i} (D(v))(w)>\sigma(v)
\}
\end{align*}
Then we have $\mathsf{win}_\exists = \mathsf{E}^{f_{\exists p}}$ (formally, this is a
consequence of Lemma~\ref{lemm:modcheck}, below).
The winning region of $\mathsf{Abelard}$ is characterized in a dual manner.
Note that nodes $v$ with $\sigma(v)=p$ correspond to model checking
modal operators
$\langle p \rangle$ which require that their argument is satisfied
with probability more than $p$ in the next step
(see Example~\ref{expl:logics}.\ref{expl:probmu}). The full probabilistic
$\mu$-calculus also has dual operators $[p]$ which state that their 
argument holds with probability at least $1-p$ in the next step; 
for brevity, we refrain from modelling these operators in this example.
\end{expl}
While the above example apparently already goes 
beyond the setting of~\cite{ChatterjeeEA18} in which standard parity
games with existential and universal branching are hardwired throughout
(e.g. in the set operator $\mathit{CPre}$ and the function
$\mathit{best}$), we note that our results cover systems of fixpoint
equations for arbitrary functions over finite lattices, 
which need not be `game-like' at
all, that is, they need not be parametrized by any graph structure or priority
and player
assignment.

We next show how to apply our results to model checking and
satisfiability checking for generalized $\mu$-calculi in the setting
of coalgebraic logic, covering, for instance,
graded~\cite{KupfermanEA02}, probabilistic~\cite{CirsteaEA11,LiuEA15},
and alternating-time~\cite{AlurEA02} $\mu$-calculi. It has been shown in
previous work~\cite{HausmannSchroder19b} that model checking for
coalgebraic $\mu$-calculi reduces to computing winning regions in a
generalized variant of parity games where the game arenas are
coalgebras instead of Kripke frames.  We proceed to recall basic
definitions and examples in universal coalgebra~\cite{Rutten00} and
the coalgebraic $\mu$-calculus~\cite{CirsteaEA11a} and then continue to
show that our main result yields new quasipolynomial-time upper bounds
for the model checking problem and improves the known exponential-time
upper bound for the satisfiability problem~\cite{HausmannSchroder19a}
of the coalgebraic $\mu$-calculus. These generic results instantiate
to new upper bounds in all concrete cases except the standard
relational $\mu$-calculus.

The abstraction principle underlying universal coalgebra is to
encapsulate system types as functors, for our present purposes on the
category of sets. Such a functor $T:\Set\to\Set$, which we fix in the
following, maps every set~$X$ to a set~$TX$, and every map $f:X\to Y$
to a map $Tf:TX\to TY$, preserving identities and composition. We
think of $TX$ as a type of structured collections over~$X$; a basic
example is the covariant powerset functor~$\Pow$, which assigns to
each set its powerset and acts on maps by taking forward
image. Systems of the intended type are then cast as
\emph{$T$-coalgebras} $(C,\xi)$ (or just~$\xi$) consisting of a
set~$C$ of \emph{states} and a \emph{transition map} $\xi:C\to T C$,
thought of as assigning to each state~$x\in C$ a structured collection
$\xi(x)\in T C$ of successors. E.g.\ a $\Pow$-coalgebra
$\xi:C\to\Pow C$ assigns to each state a set of successors, i.e.\ is a
transition system.

Following the paradigm of \emph{coalgebraic logic}~\cite{CirsteaEA11},
we fix a set~$\Lambda$ of modal operators; we interpret each
$\hearts\in\Lambda$ as \emph{predicate lifting} $\Sem{\hearts}$
for~$T$, i.e.\ a natural transformation
\begin{equation*}
  \Sem{\hearts}_X:2^X\to 2^{TX}.
\end{equation*}
Here, the index~$X$ ranges over all sets; $2^X$ denotes the set of
maps $X\to 2$ into the two-element set $2=\{\bot,\top\}$, isomorphic
to the powerset of~$X$ (i.e.\ $2^{-}$ is the \emph{contravariant
  powerset functor}; we generally keep the conversion between $2^X$
and~$\Pow(X)$ implicit); and naturality means that
$\Sem{\hearts}_X(f^{-1}[A])=(Tf)^{-1}[\Sem{\hearts}_Y(A)]$ for
$f:X\to Y$ and $A\in2^Y$. Thus, the predicate lifing $\Sem{\hearts}$
indeed lifts predicates on a base set~$X$ to predicates on the
set~$TX$. Standard examples for $T=\Pow$ are the predicate liftings
for the $\Box$ and $\Diamond$ modalities, given by
\begin{align*}
  \Sem{\Box}_X(A)&=\{B\in\Pow X\mid  B\subseteq A\}\quad\quad\text{and}
\\
  \Sem{\Diamond}_X(A)&=\{B\in\Pow X\mid B\cap A\neq\emptyset\}
\end{align*}
for $A\in PX$. Since we mean to form fixpoint logics, we need to
require that every $\Sem{\hearts}$ is \emph{monotone}, that is,
$A\subseteq B\subseteq X$ implies
$\sem{\hearts}_X(A)\subseteq\sem{\hearts}_X(B)$. To support negation,
we assume moreover that $\Lambda$ is closed under \emph{duals}, i.e.\
for each $\hearts\in\Lambda$ we have $\overline{\hearts}\in\Lambda$
such that
$\Sem{\overline{\hearts}}_X(A)=T X\setminus\Sem{\hearts}_X(X\setminus
A)$, chosen so that $\overline{\overline{\hearts}}=\hearts$ (e.g.\
$\overline\Box=\Diamond$, $\overline\Diamond=\Box)$.

Given a set $\mathsf{Var}$ of \emph{fixpoint variables}, the set of
\emph{formulae} $\phi,\psi,\dots$ of the coalgebraic $\mu$-calculus is
then defined by the grammar
\begin{align*}
\psi,\phi := \top \mid \bot \mid \psi\vee\phi\mid \psi\wedge&\,\phi\mid
\hearts\psi\mid X\mid \eta X.\psi\\
&(\hearts\in\Lambda, X\in \mathsf{Var},\eta\in\{\mu,\nu\}).
\end{align*}
Given a $T$-coalgebra $\xi:C\to T C$ and a valuation
$\sigma:\mathsf{Var}\to\Pow C$, the \emph{extension}
\begin{equation*}
  \Sem{\phi}_\sigma\subseteq C
\end{equation*}
of a formula~$\phi$ is defined recursively by
$\sem{X}_\sigma=\sigma(X)$; the expected clauses for the propositional
operators ($\Sem{\top}_\sigma=C$; $\Sem{\bot}_\sigma=\emptyset$;
$\Sem{\phi\land\psi}_\sigma=\Sem{\phi}_\sigma\cap\Sem{\psi}_\sigma$;
$\Sem{\phi\lor\psi}_\sigma=\Sem{\phi}_\sigma\cup\Sem{\psi}_\sigma$);
and
\begin{align*}
 \sem{\hearts\psi}_\sigma&=\xi^{-1}[\sem{\hearts}(\sem{\psi}_\sigma)]\\
 \sem{\mu X.\psi}_\sigma&=\mathsf{LFP}\sem{\psi}^X_\sigma \\
\sem{\nu X.\psi}_\sigma&=\mathsf{GFP}\sem{\psi}^X_\sigma
\end{align*}
where the (monotone) map $\sem{\psi}^X_\sigma:\Pow C\to\Pow C$ is
defined by $\sem{\psi}^X_\sigma (A)=\sem{\psi}_{\sigma[X\mapsto A]}$
for $A\subseteq C$, with $(\sigma[X\mapsto A])(X)=A$ and
$(\sigma[X\mapsto A])(Y)=\sigma(Y)$ for $X\neq Y$. % ; monotonicity of
% $\sem{\psi}^X_\sigma$ is clearly an invariant of the recursive
% definition.

The \emph{alternation depth} $\mathsf{ad}(\eta X.\psi)$ of a fixpoint
$\eta X.\psi$ is the depth of alternating nesting of such fixpoints in
$\psi$ that depend on $X$; we assign \emph{odd} numbers to least
fixpoints and \emph{even} numbers to greatest fixpoints. E.g.  for
$\psi=\nu X.\phi$ and $\phi=\mu Y.(p\wedge\hearts X)\vee \hearts Y$,
we have $\mathsf{ad}(\psi)=2$, $\mathsf{ad}(\phi)=1$.  For a detailed
definition of alternation depth, see
e.g.~\cite{Niwinski86}.

\begin{expl}\label{expl:logics}
  As indicated above, the standard relational
  $\mu$-calculus~\cite{Kozen83} is one example of a coalgebraic
  $\mu$-calculus, with propositional atoms treated as nullary
  modalities. Further important examples are as
  follows~\cite{SchroderPattinson09,CirsteaEA11a,SchroderVenema18}.
\begin{enumerate}[wide]
\item The \emph{graded $\mu$-calculus}~\cite{KupfermanEA02} has
    modalities $\langle b\rangle$, $[b]$, indexed over $b\in\Nat$,
    read `in more than $b$ successors' and `in all but at most~$b$
    successors', respectively. These can be interpreted over
    relational structures but it is more natural and technically more
    convenient to use \emph{multigraphs}~\cite{DAgostinoVisser02},
    i.e.\ transition systems with edge weights (\emph{multiplicities})
    in $\Nat\cup\{\infty\}$, which are coalgebras for the multiset
    functor $\Bag$ given by $\Bag X=(X\to(\Nat\cup\{\infty\}))$. % ; we
    % treat elements $\beta\in\Bag X$ as $(\Nat\cup\{\infty\})$-valued
    % discrete measures on~$X$, and in particular write
    % $\beta(A)=\sum_{x\in A}\beta(x)$ for $A\subseteq X$.  For a map
    % $f:X\to Y$, the map $\Bag(f):\Bag X\to\Bag Y$ is then given by
    % $\Bag(f)(\beta)(y)=\beta(f^{-1}[\{y\}])$.
    Over $\Bag$, we interpret $\langle b\rangle$ and $[b]$ by the
    mutually dual predicate liftings
    \begin{align*}
      \Sem{\langle b\rangle}_X(A)&=\{\beta\in\Bag X\mid\textstyle\sum_{x\in
                                   A}\beta(x)>b\}\\
      \Sem{[b]}_X(A)&=\{\beta\in\Bag X\mid\textstyle\sum_{x\in X\setminus
                      A}\beta(x)\le b\}.
    \end{align*}
    E.g.\ the formula $\nu X.\,(\phi\land\Diamond_1 X)$ says that the
    current state is the root of an infinite tree with branching
    degree at least~$2$ (counting multiplicities) on which~$\phi$
    holds everywhere.
  \item\label{expl:probmu} The (two-valued) \emph{probabilistic
      $\mu$-calculus}~\cite{CirsteaEA11a,LiuEA15} is interpreted over
    Markov chains, which are coalgebras for the \emph{discrete
      distribution functor}~$\Dist$ where
    $\Dist X=\{\beta:X\to[0,1]\mid \sum_{x\in X}\beta(x)=1\}$ is the
    set of discrete probability distributions on~$X$, represented,
    e.g., as probability mass functions $\beta:X\to[0,1]$. We
    abuse~$\beta$ to denote also the induced probability distribution,
    writing $\beta(A)=\sum_{x\in A}\beta(x)$ for $A\subseteq X$.  The
    logic has modalities $[p]$, $\langle p\rangle$ indexed over
    $p\in[0,1]\cap\Rat$, interpreted over~$\Dist$ by
    \begin{align*}
      \Sem{\langle p\rangle}_X(A)&=\{\beta\in\Dist X\mid\beta(A)>p\}\\
      \Sem{\langle p\rangle}_X(A)&=\{\beta\in\Dist X\mid\beta(X\setminus A)\le p\}.
    \end{align*}
    This example (as well as the previous one) can be extended to
    admit (monotone) polynomial inequalities among probabilities (or
    multiplicities, respectively) instead of only comparison with
    constants, allowing, e.g., for expressing probabilistic
    independence~\cite{FaginHalpern94,KupkeEA15,HausmannSchroder19a}. In
    more detail, we can introduce $n$-ary modalities $L_{p,b}$,
    $M_{p,b}$ indexed over
    polynomials~$p\in\Rat_{\ge 0}[x_1,\dots,x_n]$ and rational numbers
    $b\ge 0$, with $L_{p,b}$ interpreted by the predicate lifting
    \begin{align*}
      \sem{L_{p,b}}_X(A_1,\ldots,A_n)=\{\beta&\in\Dist X\mid\\
      & 
      p(\beta(A_1),\ldots,\beta(A_n))> b\}
    \end{align*}
    % \sem{M_{p,b}}_X(A_1,\ldots,A_n)&=\{\beta\in\Dist X\mid 
    %                                    p(\beta(\overline{A_1}),\ldots,\beta(\overline{A_n}))\leq b\}
    and $M_{p,b}$ by the corresponding dual predicate lifting. E.g.\
    the formula
    \begin{equation*}
      \nu X.\,\mu Y.\,L_{x_1x_2,0.8}(p\land X,q\lor Y)
    \end{equation*}
    says roughly that if we independently sample two successors of the
    current state, then with probability at least $0.8$, the first
    successor state will satisfy~$p$, and then~$X$ again (continuing
    indefinitely), and the second successor state will remain on a
    path where it satisfies~$Y$ again until it eventually reaches~$q$.

  % \noindent Further examples include
  % the \emph{monotone $\mu$-calculus}~\cite{EnqvistEA15}, which embeds
  % \emph{concurrent dynamic logic}~\cite{Peleg87} and Parikh's
  % \emph{game logic}~\cite{Parikh85}; and the \emph{alternating-time
  %   $\mu$-calculus (AMC)}~\cite{AlurEA02};
  % see~\cite{SchroderPattinson09,CirsteaEA09,SchroderVenema18} for more
  % details.
\item \emph{Monotone $\mu$-calculus:} The \emph{monotone
      neighbourhood functor} $\CM$ maps a set~$X$ to the set
    \begin{equation*}
      \CM X=\{\FA\in 2^{(2^X)}\mid \FA \text{ upwards closed}\}
    \end{equation*}
    of set systems over~$X$ that are upwards closed under subset
    inclusion (i.e.\ $A\in\FA$ and $A\subseteq B$ imply $B\in\FA$).
    Coalgebras for~$\CM$ are \emph{monotone neighbourhood frames} in
    the sense of Scott-Montague semantics~\cite{Chellas80}. We take
    $\Lambda=\{\Box,\Diamond\}$ and interpret~$\Box$ over~$\CM$ by the
    predicate lifting
    \begin{align*}
      \Sem{\Box}_X(A)&=\{\FA\in\CM X\mid A\in\FA\}\\
                     &=\{\FA\in\CM X\mid\exists B\in\FA.\,B\subseteq A\},
    \end{align*}
    and $\Diamond$ by the corresponding dual lifting,
    $\Sem{\Diamond}_X(A)=\{\FA\in\CM X\mid (X\setminus
    A)\notin\FA\}=\{\FA\in\CM X\mid\forall B\in\FA.\,B\cap
    A\neq\emptyset\}$.
    The arising coalgebraic $\mu$-calculus is known as the
    \emph{monotone $\mu$-calculus}~\cite{EnqvistEA15}. When we add
    propositional atoms and actions, and replace~$\CM$ with its
    subfunctor~$\CM_s$ defined by
    $\CM_sX=\{\FA\in\CM X\mid \emptyset\notin\FA\owns X\}$, whose
    coalgebras are \emph{serial} monotone neighbourhood frames, we
    arrive at the ambient fixpoint logic of \emph{concurrent dynamic
      logic}~\cite{Peleg87} and Parikh's \emph{game
      logic}~\cite{Parikh85}. In game logic, actions are understood as
    atomic games of Angel vs.\ Demon, and we read $\Box_a\phi$ as
    `Angel has strategy to enforce~$\phi$ in game~$a$'. Game logic is
    then mainly concerned with composite games, formed by the control
    operators of dynamic logic and additional ones; the semantics can
    be encoded into fixpoint definitions. For instance, the formula
    $\nu X.\,p\land\Box_a X$ says that Angel can enforce~$p$ in the
    composite game where~$a$ is played repeatedly, with Demon deciding
    when to stop.
    
  \item \emph{Alternating-time $\mu$-calculus:} Fix a
    set~$N=\{1,\dots,n\}$ of \emph{agents}. Using alternative notation
    from \emph{coalition logic}~\cite{Pauly02}, we present the
    \emph{alternating-time $\mu$-calculus (AMC)}~\cite{AlurEA02} by
    modalities $[D]$, $\langle D\rangle$ indexed over
    \emph{coalitions} $D\subseteq N$, read `$D$ can enforce' and `$D$
    cannot prevent', respectively. We define a functor $\Gm$ by
    \begin{multline*}\textstyle
      \Gm X=\{(k_1,\dots,k_n,f)\mid
      k_1,\dots,k_n\in\Nat\setminus\{0\},\\ f:\big(\textstyle\prod_{i\in N}[k_i]\big)\to X\}
    \end{multline*}
    where we write $[k]=\{1,\dots,k\}$ in this example. We understand
    $(k_1,\dots,k_n,f)\in\Gm X$ as a one-step concurrent game with
    $k_i$ available moves for agent~$i\in N$, and outcomes in~$X$
    determined by the \emph{outcome function}~$f$ from a joint choice
    of moves by all the agents. For $D\subseteq N$, we write
    $S_D=\prod_{i\in D}[k_i]$. Given joint choices $s_D\in S_D$,
    $s_{\overline{D}}\in S_{\overline{D}}$ of moves for $D$ and
    $\overline{D}=N\setminus D$ respectively, we write
    $(s_D,s_{\overline{D}})\in s_N$ for the joint move of all agents
    induced in the evident way. In this notation, we interpret the
    modalities $[D]$ over~$\Gm$ by the predicate lifting
    \begin{multline*}
      \Sem{[D]}_X(A)=\{(k_1,\dots,k_n,f)\in\Gm X\mid\\
      \exists s_D\in S_D.\,\forall s_{\overline{D}}\in S_{\overline{D}}.\,f(s_D,s_{\overline{ D}})\in A\},
    \end{multline*}
    and the modalities $\langle D\rangle$ by dualization. This
    captures exactly the semantics of the AMC: $\Gm$-coalgebras are
    precisely \emph{concurrent game structures}~\cite{AlurEA02}, i.e.\
    assign a one-step concurrent game to each state, and $[D]\phi$
    says that the agents in~$D$ have a joint move such that however
    the agents in~$\overline{D}$ move, the next state will
    satisfy~$\phi$. E.g.\ $\nu Y.\,\mu X.\,(p\land [D] Y)\lor[D] X$
    says that coalition~$D$ can enforce that~$p$ is
    satisfied infinitely often.
  \end{enumerate}
\end{expl}

\noindent We now fix a target formula $\chi$ that does not contain
free fixpoint variables, assuming w.l.o.g.\ that~$\chi$ is
\emph{clean}, i.e.\ that every fixpoint variable is bound by at most
one fixpoint operator in~$\chi$.  For a variable $x\in \mathsf{Var}$
that is bound in $\chi$, we then write $\theta(x)$ to denote
\emph{the} formula $\eta X.\psi$ that is a subformula of $\chi$.  Let
$\mathsf{Cl}(\chi)$ be the \emph{closure} (that is, the set of
subformulae) of $\chi$. We have $|\mathsf{Cl}(\chi)|\leq |\chi|$,
where $|\chi|$ denotes the number of operators or variables in $\chi$.

We proceed to recall how model checking in the coalgebraic
$\mu$-calculus is reduced to computing a nested fixpoint of a
particular function~\cite{HausmannSchroder19b}:

\begin{definition}[Coalgebraic model checking function]\upshape
  Let $\xi:C\to TC$ be a coalgebra, and $U=\mathsf{Cl}(\chi)\times C$.
  The \emph{(coalgebraic) model checking function}
  $\alpha_{\mathsf{mc}}:\Pow(U)^{k+1}\to\Pow(U)$ is given by
  putting, for $\mathbf{U}=(U_1,\ldots,U_{k+1})\in\Pow(U)^{k+1}$,
\begin{align*}
\alpha_\mathsf{mc}(\mathbf{U})=&\{(\top,x)\mid (\top,x)\in U\}\cup\\
&\{(\hearts\psi,x)\in U\mid
\xi(x)\in\sem{\hearts}\{y\mid (\psi,y)\in U_1\}\}\cup\\
&\{(\psi\vee\phi,x)\in U\mid \{(\psi,x),(\phi,x)\}\cap U_1\neq\emptyset\}\cup\\ 
&\{(\psi\wedge\phi,x)\in U\mid \{(\psi,x),(\phi,x)\}\subseteq U_1\}\cup\\
&\{(\eta X.\,\psi,x)\in U\mid (\psi,x)\in 
U_1\}\cup\\
&\{(X,x)\mid (\theta(X),x)\in
U_{\mathsf{ad}(\theta(X))+1}\}.
\end{align*}
\end{definition}

\begin{lemma}[Coalgebraic model checking~\cite{HausmannSchroder19b}]
\label{lemm:modcheck}
  Let~$\chi$ be a formula of alternation depth~$k$, $\xi:C\to TC$ a
  coalgebra, and $x\in C$ a state.  Then we have
\begin{equation*}
(\chi,x)\in\mathsf{A}^{\alpha_\mathsf{mc}}\text{ if and only if }x\in\sem{\chi}.
\end{equation*}
\end{lemma}
\noindent \noindent The \emph{one-step satisfaction problem} consists
in deciding whether $t\in\sem{\hearts} (W),$ for given $t\in T C$,
$\hearts\in\Lambda$ and $W\subseteq C$.  The time
$t(\alpha_\mathsf{mc})$ it takes to compute the model checking
function $\alpha_\mathsf{mc}$ hence depends on the time it takes to
solve the one-step satisfaction problem for the modal operators at
hand.  By Corollary~\ref{cor:qp}, we
obtain
\begin{corollary}
Model checking for coalgebraic $\mu$-calculus formulae
of alternation depth $k$ against coalgebras $\xi:C\to TC$
can be done in time
$t_2\cdot t_1$
where 
$t_2=\max(t(\alpha_\mathsf{mc}),t_1)$,
$t_1={2n^3(k+1)^2} {{\log (n(k+1))+k+2}\choose{k+1}}$, and
 $n=|\mathsf{Cl}(\chi)|\cdot|C|$.
\end{corollary}
% \begin{proof}
% Immediate since the
% set $\mathsf{A}^{\alpha_\mathsf{mc}}_k$ can -- by the proof of
%  -- be computed in time
% $s_2\cdot n\cdot k^{\lceil \log n\rceil+2}$.
% \end{proof}

\begin{expl}
  [Quasipolynomial-time model checking for graded and probabilistic
  $\mu$-calculi] In~\cite[Examples 3.2 and 3.3]{HausmannSchroder19b},
  it was shown that in the graded and probabilistic cases, the
  one-step satisfaction problem can be solved in time
  $\mathcal{O}(\mathsf{size}(\chi)\cdot|C|)$ and
  $\mathcal{O}(\mathsf{size}(C)^2\cdot|C|^3)$, respectively; here,
  $\size(\chi)$ denotes the representation size of the formula $\chi$
  and $\size(C)$ denotes the representation size of the coalgebra $C$.
  We hence obtain the following quasipolynomial upper time bounds for
  the model checking problems of the respective $\mu$-calculi, both
  with numbers coded in binary (where $t_1={2n^3(k+1)^2} {{\log (n(k+1))+k+2}\choose{k+1}}$ and $n=|\mathsf{Cl}(\chi)|\cdot|C|$):
  \begin{itemize}[itemsep=0em,parsep=0em,topsep=0.2em,partopsep=0em]
  \item for the graded $\mu$-calculus:
  $\mathcal{O}(t_1\cdot t_2)$, where $t_2=\max(\mathsf{size(\chi)}\cdot|C|,t_1)$;
  \item for the probabilistic $\mu$-calculus:
    $\mathcal{O}(t_1\cdot t_2)$, where 
$t_2=\max(\mathsf{size}(C)^2\cdot|C|^3,t_1)$.
  \end{itemize}
  Similar bounds, with slightly larger~$t_2$, are obtained for the
  respective extensions with polynomial inequalities. To the best of
  our knowledge, these bounds are new. We similarly obtain
  quasipolynomial bounds for model checking the monotone
  $\mu$-calculus and the alternating-time $\mu$-calculus. In these
  cases, the \emph{time} bounds are already
  in~\cite{HausmannSchroder19b}, via an encoding into standard parity
  games; but we emphasize again that the point of our main result
  (Corollary~\ref{cor:qp}) is not so much the time bound but rather
  the quasipolynomial bound on the number of iterations -- in this
  case, we obtain that the fixpoint can be computed with
  quasipolynomially many calls to the one-step satisfaction problem
  (which at least for the alternating-time case seems also
  algorithmically preferable to an encoding in parity games with many
  additional states).
\end{expl}

\noindent We now consider satisfiability checking for the coalgebraic
$\mu$-calculus, which also reduces to the computation of a nested
fixpoints of a certain function~\cite{HausmannSchroder19a}. We recall
the essential notions that are required to define this function;
see~\cite{HausmannSchroder19a} for details of the construction.  We
fix a target formula $\chi$ of size $n$ and alternation depth $k$, to
be checked for satisfiability. One then has a deterministic parity
automaton that accepts
precisely the \emph{good branches} in tableaux representing
prospective models of~$\chi$, i.e.\ the ones not containing infinite
deferrals of least fixpoints (which represent eventualities). We work
with parity automata in which priorities are assigned to the
\emph{transitions} (rather than the states); our automaton thus has
the form $(\detcarrier,\Sigma,\delta,\beta)$ where $\detcarrier$ is
the set of states; $\Sigma$ is the alphabet (designed to allow
identifying manipulations of formulae happening in the
transitions);~$\delta$ is the transition function; and~$\beta$ assigns
priorities to transitions. Since the automaton is deterministic, we
can take~$\beta$ to be a function
$\detcarrier\times\Sigma\to\Nat$. Recall that such a parity automaton
accepts an infinite word~$w$ if and only if it has a run for~$w$ in which the
highest priority that occurs infinitely often is even. We have
$|\detcarrier|\in\mathcal{O}(2^{\mathcal{O}(nk\log n)})$, and nodes $v\in\detcarrier$
are labelled with sets $l(v)$ of formulae. We denote the set of nodes
whose labels contain some propositional formula by
$\mathsf{prestates}$ and the set of nodes whose labels contain only
modal formulae by $\mathsf{states}$; for $v\in\mathsf{prestates}$,
$\psi_v$ is a fixed propositional formula from the label of~$v$. The
transition function $\delta$ tracks sets of formulae according to the
logical manipulations described by a given letter from~$\Sigma$.
Besides letters identifying propositional transformations,~$\Sigma$
contains sets of modal formulae describing modal steps; we write
$\mathsf{selections}\subseteq\Sigma$ for the set of these letters.
\begin{definition}[Coalgebraic satisfiability checking function~\cite{HausmannSchroder19a}] For sets
  $U\subseteq \detcarrier$ and
  $\mathbf{U}=(U_1,\ldots, U_{\prios}) \in\Pow(U)^{\prios}$, we put
  \begin{align*}
    \alpha_{\mathsf{sat}}(\mathbf{U})=&\{v\in \mathsf{prestates}\mid \\
                                      &\qquad\exists b\in \{0,1\}.\,
                                        \delta(v,(\psi_v,b))\in U_{\detprio(v,(\psi_v,b))}\}\cup\\
                                      &\{v\in \mathsf{states}\mid T(\textstyle\bigcup_{1\leq i\leq \prios} U_i(v))\cap\sem{l(v)}_1\neq\emptyset\}
  \end{align*}
  where $\detprio(v,(\psi_v,b))$ abbreviates $\detprio(v,(\psi_v,b),\delta(v,(\psi_v,b)))$ and where
  \begin{align*}
    U_i(v)=\{l(u)\mid &\,u\in X_i,\exists\kappa\in\mathsf{selections}.\\
                      &
                        \delta(v,\kappa)=u, \detprio(v,\kappa,u)=i\}.
  \end{align*}
\end{definition}
\noindent The \emph{one-step satisfiability problem} is to decide
whether
\begin{equation*}
  T(\textstyle\bigcup_{1\leq i\leq \prios}
  U_i(v))\cap\sem{l(v)}_1\neq\emptyset
\end{equation*}
for given~$U$, $v$.  Hence checking whether some
$v\in\mathsf{prestates}$ is contained in $\alpha_{\mathsf{sat}}(\mathbf{U})$
for given $\mathbf{U}$ is an instance of the one-step satisfiability problem.

\begin{lemma}[Fixpoint characterization of satisfiability \cite{HausmannSchroder19a}]
  In the above notation,
  \begin{equation*}
    v_0\in\mathsf{A}^{\alpha_\mathsf{sat}}\text{ if and only if }\chi\text{ is satisfiable}.
  \end{equation*}
\end{lemma}

\begin{corollary}
If the one-step satisfiability problem of a coalgebraic
logic can be solved in time $2^{\mathcal{O}({nk\log n})}$, then
the satisfiability problem of the $\mu$-calculus over this
logic can be solved in time $2^{\mathcal{O}({nk\log n})}$ as well.
\end{corollary}
\begin{proof}
By the previous Lemma, it suffices to show that
$\mathsf{A}^{\alpha_\mathsf{sat}}$ can be computed in time
$2^{\mathcal{O}({nk\log n})}$. Since we have $\prios<\log{|\detcarrier|}$, $\mathsf{A}^{\alpha_\mathsf{sat}}$ can
-- by Corollary~\ref{cor:fewpriopoly} -- be computed in time 
$\mathcal{O}(t(\alpha_\mathsf{sat})\cdot (|\detcarrier|(\prios+1))^8)$, where
$t(\alpha_\mathsf{sat})$ denotes the maximum of $(|\detcarrier|(\prios+1))^8$ and the time it takes to compute $\alpha_\mathsf{sat}$;
by assumption, $\alpha_\mathsf{sat}$ can be computed in time
$2^{\mathcal{O}({nk\log n})}$ so that we have $t(\alpha_\mathsf{sat})\in 2^{\mathcal{O}({nk\log n})}$.
\qed
\end{proof}

\begin{expl} It has been shown (e.g. in \cite{HausmannSchroder19a})
  that the one-step satisfiability problems of all logics from
  Example~\ref{expl:logics} can be solved in time
  $2^{\mathcal{O}({nk\log n})}$. Hence we obtain an upper bound
  $2^{\mathcal{O}({nk\log n})}$ for the satisfiability problems of all
  these logics, in particular including the monotone $\mu$-calculus,
  the alternating-time $\mu$-calculus, the graded $\mu$-calculus and
  the (two-valued) probabilistic $\mu$-calculus, even when the latter
  two are extended with (monotone) polynomial inequalities. This
  improves on the best previous bounds in all cases.

\end{expl}

\end{document}

%%% Local Variables:
%%% mode: latex
%%% TeX-master: t
%%% End: